\documentclass[12pt, letterpaper]{article}
\usepackage[utf8]{inputenc}
\usepackage[margin=1in]{geometry}

\usepackage{comment}

\usepackage{graphicx}
\usepackage{subcaption}

\usepackage{mathtools}
\usepackage{amsthm}
\usepackage{bm}

\usepackage[usenames]{xcolor}

\newtheorem{theorem}{Theorem}
\newtheorem{lemma}{Lemma}

\theoremstyle{definition}

\usepackage{enumitem}

\usepackage{algorithm,algcompatible}

\usepackage[onehalfspacing]{setspace}

\usepackage[round]{natbib}

\newcommand{\dual}{\textsc{D}}
\newcommand{\alg}{\textsc{A}}

\newcommand{\type}{\tau}

\newcommand{\chosen}{\textrm{\em matched}}
\newcommand{\notchosen}{\textrm{\em unmatched}}
\newcommand{\nonadapt}{\textrm{\em unknown}}

\newcommand{\sender}{\textrm{\em sender}}
\newcommand{\receiver}{\textrm{\em receiver}}

\newcommand{\kmax}{k_{\max}}

\newcommand{\exante}{\text{\em ex-ante}}
\newcommand{\expost}{\text{\em ex-post}}

\title{Understanding Zadimoghaddam's Edge-weighted Online Matching Algorithm: Weighted Case}

\author{
Zhiyi Huang%
\thanks{The University of Hong Kong. Email: zhiyi@cs.hku.hk}
}

\date{Last Update: October 2019}

\begin{document}

\maketitle

\begin{abstract}
    This article presents a simplification of Zadimoghaddam's algorithm for the edge-weighted online bipartite matching problem, under the online primal dual framework.
    In doing so, we obtain an improved competitive ratio of $0.514$.
    We first combine the online correlated selection (OCS), an ingredient distilled from \citet{Zadimoghaddam/arXiv/2017} by \citet{HuangT/arXiv/2019}, and an interpretation of the edge-weighted online bipartite matching problem by \citet{DevanurHKMY/TEAC/2016} which we will refer to as the complementary cumulative distribution function (CCDF) viewpoint, to derive an online primal dual algorithm that is $0.505$-competitive.
    Then, we design an improved OCS which gives the $0.514$ ratio.
\end{abstract}

\section{Introduction}
\label{sec:intro}

We consider the edge-weighted online bipartite matching problem.
Let there be a bipartite graph $G = (L, R, E)$, where $L$ and $R$ are the sets of vertices on the left-hand-side (LHS) and right-hand-side (RHS) respectively, and $E \subseteq L \times R$ is the set of edges.
Further, every edge $(i, j)$ is associated with a nonnegative weight $w_{ij} \ge 0$.
We will assume without loss of generality (wlog) that $E = L \times R$ by assigning zero weights to the missing edges.

The LHS is offline in that these vertices are given upfront.
The vertices on the RHS, however, arrive online one at a time.
On the arrival of an online vertex $j \in R$, its incident edges and their weights are revealed to the algorithm, who must then irrevocably match $j$ to an offline neighbor.
Each offline vertex may be matched any number of times, yet only the weight of the heaviest edge counts towards the objective.
This is equivalent to allowing a matched offline vertex $i$, say, to $j'$, to be rematched to a new online vertex $j$ whose edge weight $w_{ij}$ is larger than $w_{ij'}$, disposing $j'$ and edge $(i, j')$ for free. 
Therefore, it is also known as the \emph{free disposal} model.
We seek to maximize the total weight of the matching.

On the one hand, there is a simple $0.5$-competitive greedy algorithm, which matches each online vertex to the offline neighbor that gives the largest marginal increase in the objective, i.e., the weight of the new edge minus that of the old one, if any.
On the other hand, the problem generalizes the unweighted online bipartite matching problem by \citet{KarpVV/STOC/1990} and the vertex-weighted version by \citet{AggarwalGKM/SODA/2011} and therefore, no algorithm can get a competitive ratio better than the optimal $1 - \frac{1}{e}$ ratio of the two special cases.

\citet{FeldmanKMMP/WINE/2009} give an online algorithm with a competitive ratio strictly better than $0.5$, for a variant of the edge-weighted problem which counts the heaviest $n > 1$ edges matched to each offline vertex;
it can be interpreted as a fractional version of the problem.
The analysis is subsequently simplified by \citet{DevanurHKMY/TEAC/2016} under the online primal dual framework.

\citet{Zadimoghaddam/arXiv/2017} presents an algorithm for the original problem and gives a competitive ratio of $0.5018$.
In an effort to provide a more accessible algorithm and analysis, \citet{HuangT/arXiv/2019} distill a key ingredient from the algorithm of Zadimoghaddam, which they call the online correlatd selection (OCS), and explain how to use the OCS to obtain a simplified $0.505$-competitive algorithm in the unweighted case.

This article continues the work of \citet{HuangT/arXiv/2019}. 
We give a simplified algorithm for the edge-weighted online bipartite matching problem with an improved competitive ratio, under the online primal dual framework.
The main result is the following theorem.

\begin{theorem}
    \label{thm:main}
    There is a $0.514$-competitive online primal dual algorithm for the edge-weighted online bipartite matching problem, in the free-disposal model.
\end{theorem}

The algorithm relies on a number of ingredients, including the aforementioned OCS by \citet{HuangT/arXiv/2019}, an interpretation of the edge-weighted online bipartite matching problem by \citet{DevanurHKMY/TEAC/2016}, which we will refer to as the \emph{complementary cumulative distribution function (CCDF) viewpoint}, and the online primal dual framework and its instantiation in the CCDF viewpoint.
The online primal dual framework has been extensively used in the design and analysis of online matching and more generally in online algorithms.
See, e.g., \citet{DevanurJK/SODA/2013} for a comprehensive discussion on this topic, and also \citet{BuchbinderJN/ESA/2007}, \citet{DevanurJ/STOC/2012}, \citet{DevanurHKMY/TEAC/2016}, \citet{HuangKTWZZ/STOC/2018}, \citet{HuangPTTWZ/SODA/2019}, \citet{HuangTWZ/ICALP/2018}, \citet{WangW/ICALP/2015} for a non-exhaustive list of applications of the online primal dual framework in online matching problems.

Section~\ref{sec:prelim} introduces the above technical preliminaries.
Section~\ref{sec:algorithm} demonstrates how to combine these ingredients to obtain an online primal dual algorithm which not only is simpler than the original algorithm by \citet{Zadimoghaddam/arXiv/2017}, but also achieves a better competitive ratio of $0.505$, matching the ratio in the unweighted case by \citet{HuangT/arXiv/2019}.
Finally, Section~\ref{sec:improved-OCS} presents an improved version of the OCS, which leads to the competitive ratio in the main theorem.

\section{Preliminaries}
\label{sec:prelim}


\subsection{Online Correlated Selection}
\label{sec:OCS}

The OCS is an online algorithm that takes a sequence of pairs of candidates as input, one pair at a time. 
On receiving each pair, it immediately chooses one of the two candidates.
The goal is to ensure that (1) the marginal distribution of each round is a uniform distribution over the two candidates, and (2) for any fixed candidate, the randomness in the rounds in which it is involved is negatively correlated, such that after $k \ge 2$ rounds the probability that it is never chosen is strictly less than $2^{-k}$.

We include a formal definition of the OCS by \citet{HuangT/arXiv/2019} as Algorithm~\ref{alg:OCS} for completeness, and the formal statements of its properties in Lemma~\ref{lem:OCS} and Lemma~\ref{lem:OCS-generalized}.
Readers are referred to \citet{HuangT/arXiv/2019} for a detailed explanation of the algorithm and the proofs of Lemma~\ref{lem:OCS} and Lemma~\ref{lem:OCS-generalized}.

\begin{lemma}
    \label{lem:OCS}
    For any fixed sequence of pairs of candidates, any fixed candidate $i$, and any $k \ge 0$, the OCS in Algorithm~\ref{alg:OCS} ensures that with probability at least $1 - 2^{-k} \cdot f_k$, $i$ is chosen in at least one of the first $k$ rounds in which it is involved, where $f_k$ is defined recursively as: 
    \begin{equation}
        \label{eqn:OCS-recurrence}
        f_k = 
        \begin{cases}
            1 & k = 0, 1 \\
            f_{k-1} - \frac{1}{16} f_{k-2} & k \ge 2
        \end{cases}
    \end{equation}
\end{lemma}

We say that $j_1 < j_2 < \dots < j_k$ is a sequence of consecutive rounds in which $i$ is involved if $i$ is a candidate in these rounds, but in no other rounds in between.

\begin{lemma}
    \label{lem:OCS-generalized}
    For any fixed sequence of pairs of candidates, any fixed candidate $i$, and any disjoint sequences of consecutive rounds of lengths $k_1, k_2, \dots, k_m \ge 1$ in which $i$ is involved, the OCS in Algorithm~\ref{alg:OCS} ensures that $i$ is chosen in at least one of the rounds with probability at least:
    \[
        1 - \prod_{\ell=1}^m 2^{-k_\ell} \cdot f_{k_\ell} 
        ~.
    \]
\end{lemma}

By the recurrence of $f_k$, we have $f_k \le f_{k-1}$, which further implies:
\[
    f_k = f_{k-1} - \tfrac{1}{16} f_{k-2} \le \big( 1 - \tfrac{1}{16} \big) f_{k-1}
    ~.
\]

Therefore, for any $k \ge 1$, we have:
\[
    f_k \le \big( 1 - \tfrac{1}{16} \big)^{k-1}
\]

As a corollary, we have the following lemma, which will be used in the analysis.

\begin{lemma}
    \label{lem:OCS-gamma}
    For any fixed sequence of pairs of candidates, any fixed candidate $i$, and any disjoint consecutive sequences of rounds of lengths $k_1, k_2, \dots, k_m \ge 1$ in which $i$ is involved, the OCS in Algorithm~\ref{alg:OCS} ensures that, for $\gamma = \tfrac{1}{16}$, $i$ is chosen in at least one of the rounds with probability at least:
    \[
        1 - \prod_{\ell=1}^m 2^{-k_\ell} \big( 1 - \gamma \big)^{k_\ell-1}
        ~.
    \]
\end{lemma}

More generally, for any $0 \le \gamma \le 1$, we say that an online algorithm is a $\gamma$-OCS if it satisfies the above lemma for the corresponding value of $\gamma$.

\begin{algorithm}[t]
    \caption{Online Correlated Selection (OCS)}
    \label{alg:OCS}
    \begin{algorithmic}
        \medskip
        \STATEx \textbf{State variables:}
        \begin{itemize}
            \item $\type_i \in \big\{ \chosen, \notchosen, \nonadapt \big\}$ for each offline vertex $i \in L$; 
            initially, let $\type_i = \nonadapt$.
        \end{itemize}
        \STATEx \textbf{On receiving $2$ candidate offline vertices $i_1$ and $i_2$ (for an online vertex $j \in R$):}
        \begin{enumerate}
            \item With probability $\frac{1}{2}$, let it be an \emph{oblivious step}:
            \begin{enumerate}
                \item Draw $\ell, m \in \{1, 2\}$ uniformly at random.
                \item Let $\type_{i_{-m}} = \nonadapt$.
                \item If $m = \ell$, let $\type_{i_m} = \chosen$; otherwise, let $\type_{i_m} = \notchosen$.
                %
                %
            \end{enumerate}
            \item Otherwise (i.e., with probability $\frac{1}{2}$), let it be an \emph{adaptive step}:
            \begin{enumerate}
                \item Draw $m \in \{1, 2\}$ uniformly at random.
                \item If $\type_{i_m} = \chosen$, let $\ell = -m$; \newline
                if $\type_{i_m} = \notchosen$, let $\ell = m$; \newline
                if $\type_{i_m} = \nonadapt$, draw $\ell \in \{1, 2\}$ uniformly at random.
                \item Let $\type_{i_1} = \type_{i_2} = \nonadapt$.
            \end{enumerate}
            \item Return $i_\ell$.
        \end{enumerate}
        \smallskip
    \end{algorithmic}
\end{algorithm}

\subsection{CCDF Viewpoint}
\label{sec:histogram-view}

For any offline vertex $i \in L$, and any weight-level $w > 0$, let $y_i(w)$ be the probability that $i$ is matched to at least one online vertex $j$ whose $w_{ij} \ge w$.
Then, $y_i(w)$ is a non-increasing function of $w$ which takes values between $0$ and $1$ because it is exactly the complementary cumulative distribution function of the weight of the heaviest edge matched to $i$.
We remark that $y_i(w)$ is a step function with polynomially many pieces because the number of pieces is at most the number of different weights of its incident edges.
Hence, an online algorithm can maintain this information in polynomial time.

The expected weight of the heaviest edge matched to $i$ is then equal to the area under function $y_i(w)$, i.e.:
\begin{equation}
    \label{eqn:histogram-expected-weight}
    \int_0^\infty y_i(w) dw
    ~.
\end{equation}

See Figure~\ref{fig:histogram} for an example.
Suppose an offline vertex $i$ has $4$ online neighbors $j_1$, $j_2$, $j_3$, and $j_4$, with edge weights $w_1 < w_2 < w_3 < w_4$.
Further suppose that $j_1$ is matched to $i$ with certainty, while $j_2$, $j_3$, and $j_4$ each has some probability of being matched to $i$.
The latter events may or may not be correlated.
Then, we have:
\[
y_i(w) = \begin{cases}
1 & 0 < w \le w_1 \\
\Pr \big[ \text{at least one of $j_2$, $j_3$, and $j_4$ is matched to $i$} \big]
& w_1 < w \le w_2 \\
\Pr \big[ \text{at least one of $j_3$ and $j_4$ is matched to $i$} \big] & w_2 < w \le w_3 \\
\Pr \big[ \textrm{$j_4$ is matched to $i$} \big] & w_3 < w \le w_4 \\
0 & w > w_4
\end{cases}
\]

Next, suppose a new neighbor arrives whose edge weight is also $w_3$.
Then, the values of $y_i(w)$ may increase for any $w \le w_3$ according to matching decision by the algorithm.
See Figure~\ref{fig:histogram-increment} for an example.
The total area of the shaded regions is the increment in the expected weight of the heaviest edge matched to $i$.

\begin{figure}[t]
    \centering
    \begin{subfigure}{.45\textwidth}
    \includegraphics[width=\textwidth]{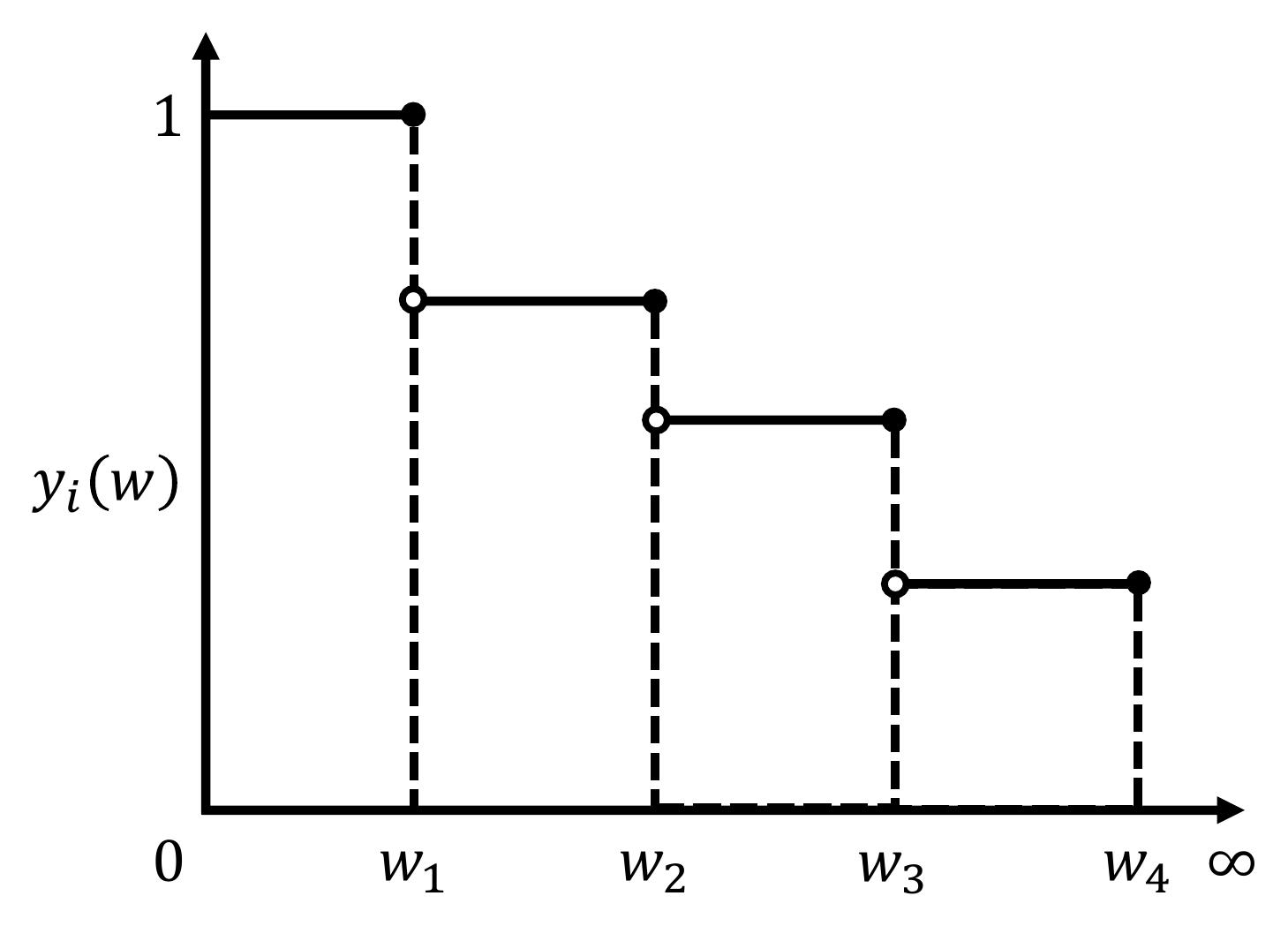}
    \caption{CCDF of a vertex $i$}
    \label{fig:histogram}
    \end{subfigure}
    \begin{subfigure}{.45\textwidth}
    \includegraphics[width=\textwidth]{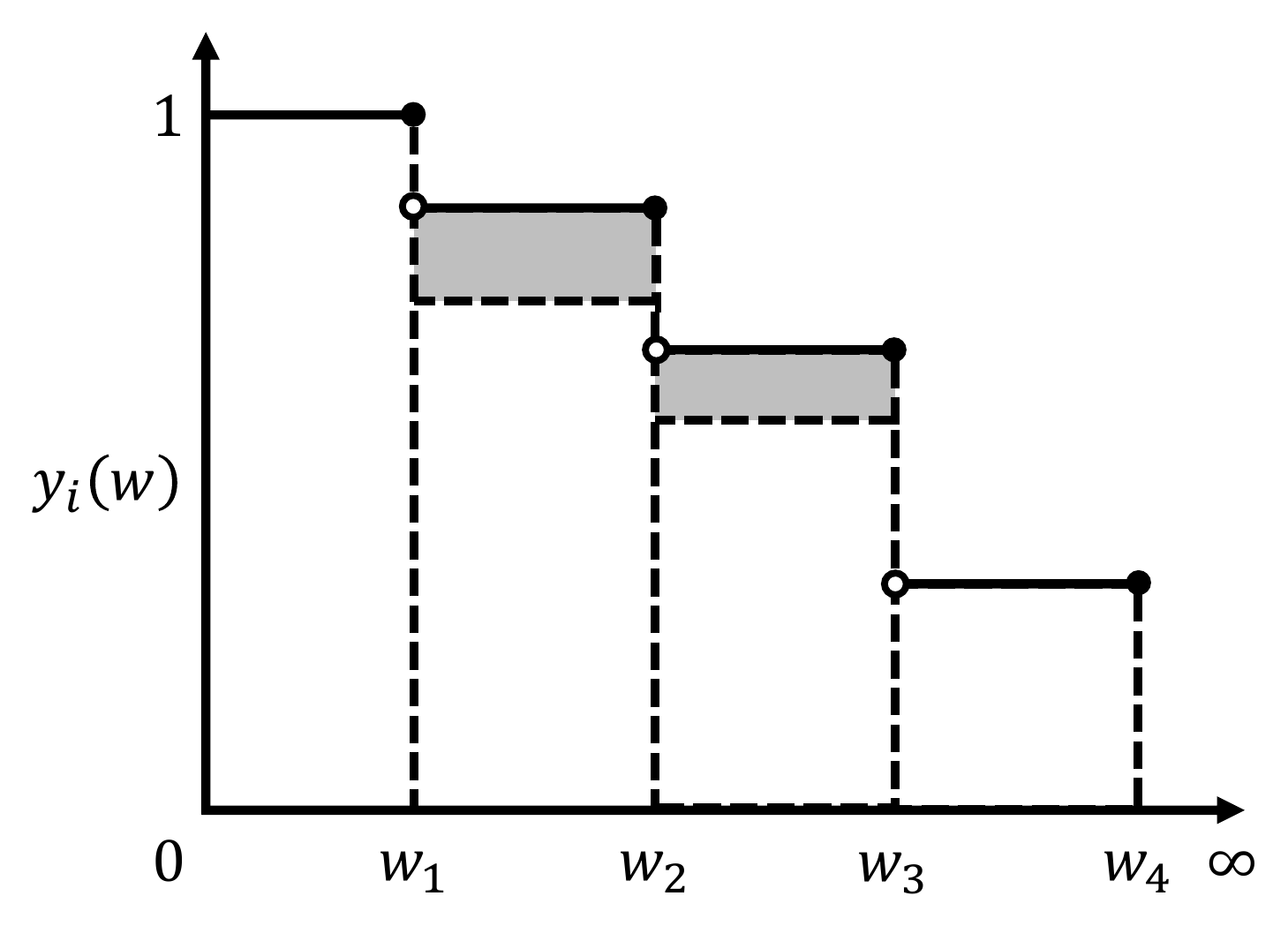}
    \caption{Updates in the CCDF viewpoint}
    \label{fig:histogram-increment}
    \end{subfigure}
    \caption{Complementary cumulative distribution function viewpoint}
\end{figure}

\subsection{Online Primal Dual}

Consider the following matching linear program (LP), in which $x_{ij}$ may be interpreted as the probability that $(i, j)$ is the heaviest edge matched to $i$:
\begin{align*}
    \textrm{maximize} \quad 
    & 
    \sum_{i \in L} \sum_{j \in R} w_{ij} x_{ij} \\
    \textrm{subject to} \quad
    &
    \sum_{j \in R} x_{ij} \le 1 && \forall i \in L \\
    &
    \sum_{i \in L} x_{ij} \le 1 && \forall j \in R \\
    &
    x_{ij} \ge 0 && \forall i \in L, \forall j \in R
\end{align*}

Further, consider its dual LP:
\begin{align*}
    \textrm{minimize} \quad
    &
    \sum_{i \in L} \alpha_i + \sum_{j \in R} \beta_j \\
    \textrm{subject to} \quad
    &
    \alpha_i + \beta_j \ge w_{ij} && \forall i \in L, \forall j \in R \\[2ex]
    &
    \alpha_i, \beta_j \ge 0 && \forall i \in L, \forall j \in R
\end{align*}

Let $\alg$ denote the expected total weight of the matching by the algorithm.
Let $\dual$ denote the dual objective.
The online primal dual framework mains not only a matching but also a dual assignment at all time subject to the conditions summarized below.

\clearpage

\begin{lemma}
    \label{lem:online-primal-dual}
    Suppose for some $0 < \Gamma \le 1$, the following conditions hold at all time:
    \begin{enumerate}
        \item (Objectives) $\alg \ge \dual$.
        \item (Approximate Dual Feasibility) For any $i \in L$ and any $j \in R$:
        \[
            \alpha_i + \beta_j \ge \Gamma \cdot w_{ij} ~.
        \]
    \end{enumerate}
    Then, the algorithm is $\Gamma$-competitive.
\end{lemma}

\begin{proof}
    By the second condition, $\Gamma^{-1} \alpha_i$'s and $\Gamma^{-1} \beta_j$'s form a feasible dual assignment, whose objective is $\Gamma^{-1} \dual$.
    By weak duality of LP, the objective of any feasible dual is an upper bound of the optimal.
    That is, $\dual$ is at least a $\Gamma$ fraction of the optimal.
    Therefore, together with the first condition we get the stated competitive ratio.
\end{proof}

\paragraph{Online Primal Dual in the CCDF Viewpoint.}
In light of the CCDF viewpoint, for any offline vertex $i \in L$, and any weight-level $w > 0$, we introduce a new variable $\alpha_i(w)$; 
we maintain $\alpha_i$ via these new variables:
\begin{equation}
    \label{eqn:alpha-by-weight} 
    \alpha_i = \int_0^\infty \alpha_i(w) dw ~. 
\end{equation}

Therefore, we rephrase Lemma~\ref{lem:online-primal-dual} in the CCDF viewpoint as follows.

\begin{lemma}
    \label{lem:online-primal-dual-rephrase}
    Suppose for some $0 < \Gamma \le 1$, the following conditions hold at all time:
    \begin{enumerate}
        \item (Objectives) $\alg \ge \dual$.
        \item (Approximate Dual Feasibility) For any $i \in L$ and any $j \in R$:
        \[
            \int_0^{\infty} \alpha_i(w) dw + \beta_j \ge \Gamma \cdot w_{ij} ~.
        \]
    \end{enumerate}
    Then, the algorithm is $\Gamma$-competitive.
\end{lemma}

\section{Online Primal Dual Algorithm}
\label{sec:algorithm}

This section gives an online primal dual algorithm for the edge-weighted online bipartite matching problem, utilizing the OCS as the main ingredient.
Importantly, the algorithm and its analysis takes a general $\gamma$-OCS as a blackbox which satisfies Lemma~\ref{lem:OCS-gamma}.
The competitive ratio depends on the value of $\gamma$.  
With the $\tfrac{1}{16}$-OCS by \citet{HuangT/arXiv/2019}, we show a competitive ratio of $0.505$.
We will further combine them with an improved OCS to get the competitive ratio in Theorem~\ref{thm:main} in the next section.

\subsection{Algorithm}

At a high-level, the algorithm works as follows.
On the arrival of an online vertex $j \in R$, it makes one of the following three kinds of decisions.
The first one is \emph{randomized}: it picks a pair of offline vertices $i_1, i_2 \in L$, lets the OCS select one of them, and matches $j$ to it.
The second one is \emph{deterministic}: it picks one offline vertex $i^* \in L$, and matches $j$ to it.
Finally, it may also leave $j$ \emph{unmatched}.

How does the algorithm decide whether it is a randomized, deterministic, or unmatched round?
How does it pick the candidate offline vertices?
On the arrival of an online vertex $j$, for every offline vertex $i$, the algorithm calculates how much $\beta_j$ would gain if $j$ is matched to $i$ in a deterministic round, denoted as $\Delta_i^D \beta_j$, and similar $\Delta_i^R \beta_j$ for a randomized round.
Then, it finds $i^*$ with the maximum $\Delta_i^D \beta_j$, and $i_1, i_2$ with the maximum $\Delta_i^R \beta_j$.
If both $\Delta_{i_1}^R \beta_j + \Delta_{i_2}^R \beta_j$ and $\Delta_{i^*}^D \beta_j$ are negative, leave $j$ unmatched. 
If $\Delta_{i_1}^R \beta_j + \Delta_{i_2}^R \beta_j$ is nonnegative and greater than $\Delta_{i^*}^D \beta_j$, match $j$ in a randomized round with $i_1$ and $i_2$ as the candidates.
Finally, if $\Delta_{i^*}^D \beta_j$ is nonnegative and greater than $\Delta_{i_1}^R \beta_j + \Delta_{i_2}^R \beta_j$, match $j$ to $i^*$ in a deterministic round.

It remains to explain how to calculate $\Delta_i^D \beta_j$ and $\Delta_i^R \beta_j$.
To this end, for any offline vertex $i \in L$ and any weight-level $w > 0$, let $k_i(w)$ be the number of randomized rounds in which $i$ is a candidate and the corresponding edge weight is at least $w$.
The increments in the dual variables are determined by the values of $k_i(w)$'s.

More concretely, the algorithm relies on the following parameters. 
Their values will be determined later in the analysis through solving an LP to optimize the competitive ratio.
\begin{itemize}
    \item $a(k)$ : Amortized increment in $\alpha_i(w)$ if $i$ is chosen as one of the two candidates in a randomized round in which its edge weight is at least $w$, and $k_i(w) = k$.
    \item $b(k)$ : Increment in $\beta_j$ due to an offline vertex $i$, at weight-level $w \le w_{ij}$, if $j$ is matched in a randomized round with $i$ as one of the two candidates, and $k_i(w) = k$.
\end{itemize}

We shall interpret $a(k)$'s and $b(k)$'s as a gain splitting rule as follows.
If $i$ is chosen as one of the two candidates to be matched to $j$ in a randomized round, the increase in the expected weight of the heaviest edge matched to $i$ is equal to the integration of $y_i(w)$'s increments, for $0 < w \le w_{ij}$, which can be related to the values of $k_i(w)$'s.
We will split the gain due to the increment of $y_i(w)$ into two parts, $a\big(k_i(w)\big)$ and $b\big(k_i(w)\big)$, such that the former goes to $\alpha_i(w)$ and the latter goes to $\beta_j$.

\begin{algorithm}[t]
    \caption{Online Primal Dual Edge-weighted Bipartite Matching}
    \label{alg:primal-dual}
    \begin{algorithmic}
        \medskip
        \STATEx \textbf{State variables:}
        \begin{itemize}
            \item $k_i(w) \ge 0$ : The number of randomized rounds in which $i$ is a candidate, and its edge weight is at least $w$;
            $k_i(w) = \infty$ if it has been chosen in a deterministic round in which its edge weight at least $w$.
        \end{itemize}
        \smallskip
        \STATEx \textbf{On the arrival of an online vertex $j \in R$:}
        \begin{enumerate}
            \item For any offline vertex $i \in L$, compute $\Delta_i^R \beta_j$ and $\Delta_i^D \beta_j$ according to Eqn.~\eqref{eqn:beta-increment-randomized} and \eqref{eqn:beta-increment-deterministic}.
            \item Find $i_1, i_2$ with the maximum $\Delta_i^R \beta_j$.
            \item Find $i^*$ with the maximum $\Delta_i^D \beta_j$.
            \item If $0 > \Delta_{i_1}^R \beta_j + \Delta_{i_2}^R \beta_j$ and $\Delta_{i^*}^D \beta_j$, leave $j$ unmatched.
            \hspace*{\fill}
            \textbf{(unmatched)}
            \item If $\Delta_{i_1}^R \beta_j + \Delta_{i_2}^R \beta_j \ge \Delta_{i^*}^D \beta_j$ and $0$, let the OCS pick one of $i_1$ and $i_2$.
            \hspace*{\fill}
            \textbf{(randomized)}
        \item If $\Delta_{i^*}^D \beta_j > \Delta_{i_1}^R \beta_j$ and $\ge 0$, match $j$ to $i^*$.
            \hspace*{\fill}
            \textbf{(deterministic)}
            \item Update $k_i(w)$'s accordingly.
        \end{enumerate}
    \end{algorithmic}
    \smallskip
\end{algorithm}

Indeed, we will show in the next subsection how to maintain the following invariant related to the dual variables $\alpha_i(w)$'s of the offline vertices:
\begin{equation}
    \label{eqn:invariant-alpha}
    \alpha_i(w) \ge \sum_{0 \le k < k_i(w)} a(k)
    ~.
\end{equation}

Further, define $\Delta^R_i \beta_j$ to be:
\begin{equation}
    \label{eqn:beta-increment-randomized}
    \Delta_i^R \beta_j = \int_0^{w_{ij}} b\big(k_i(w)\big) dw - \frac{1}{2} \int_{w_{ij}}^{\infty} \sum_{0 \le \ell < k_i(w)} a(\ell) dw 
    ~.
\end{equation}

The first term follows by the aforementioned interpretation of $b(k)$'s;
this would be the only term in the unweighted case.
The second term is designed to cancel out the extra help we get from $\alpha_i(w)$'s at weight-levels $w > w_{ij}$ in terms of satisfying approximate dual feasibility between $i$ and $j$.
Concretely, the choice of $b(k)$'s shall ensure approximate dual feasibility between $i$ and $j$, if $j$ is matched in a randomized round with two candidates at least as good as $i$, i.e.:
\[
    \int_0^{\infty} \alpha_i(w) dw + 2 \cdot \Delta_i^R \beta_j \ge \Gamma \cdot w_{ij}
    ~.
\]

By the lower bound of $\alpha_i(w)$'s in Eqn.~\eqref{eqn:invariant-alpha}, and the value of $\Delta_i^R \beta_j$ in Eqn.~\eqref{eqn:beta-increment-randomized}, it becomes an inequality independent of the weight-levels $w > w_{ij}$:
\[
    \int_0^{w_{ij}} \sum_{0 \le k < k_i(w)} a(k) dw + 2 \int_0^{w_{ij}} b\big(k_i(w)\big) dw \ge \Gamma \cdot w_{ij}
    ~.
\]

Finally, for some $1 < \kappa < 2$, define the value of $\Delta_i^D \beta_j$ to be:
\begin{equation}
    \label{eqn:beta-increment-deterministic}
    \begin{aligned}
        \Delta_i^D \beta_j 
        &
        = \kappa \cdot \Delta_i^R \beta_j \\[1ex]
        &
        = \kappa \int_0^{w_{ij}} b\big(k_i(w)\big) dw - \frac{\kappa}{2} \int_{w_{ij}}^{\infty} \sum_{0 \le \ell < k_i(w)} a(\ell) dw 
        ~.
    \end{aligned}
\end{equation}

For concreteness, readers may think of $\kappa = \frac{3}{2}$.
Nonetheless, we will show that the final competitive ratio is insensitive to the choice of $\kappa$, so long as it is neither too close to $1$, nor to $2$.
On the one hand, $\kappa > 1$ ensures that if the algorithm chooses a randomized round with offline vertex $i$ and another vertex as the candidates, the contribution from the other vertex to $\beta_j$ must be at least a $\kappa - 1$ fraction of what $i$ offers; 
otherwise, the algorithm would have preferred a deterministic round with $i$ alone.
On the other hand, we let $\kappa < 2$ because otherwise a randomized round is always inferior to a deterministic one with one of the candidate, whichever makes a larger offer.

See Algorithm~\ref{alg:primal-dual} for the formal definition of the algorithm.

\subsection{Online Primal and Dual Analysis}
\label{sec:primal-dual-analysis}
By the discussion in the previous subsection, it shall be clear that $\beta_j$ is: 
\begin{equation}
    \label{eqn:update-beta}
    \beta_j = \begin{cases}
        \Delta_{i^*}^D \beta_j & \text{in a deterministic round,} \\
        \Delta_{i_1}^R \beta_j + \Delta_{i_2}^R \beta_j & \text{in a randomized round, and} \\
        0 & \text{in an unmatched round}
    \end{cases}
\end{equation}

Next, we explain how to account for the objective by the algorithm, and the updates of $\alpha_i(w)$'s, i.e., the dual variables related to offline vertices.

\subsubsection{Lower Bound of the Algorithm Objective}

Recall that by the CCDF viewpoint:
\[
    \alg = \sum_{i \in L} \int_0^\infty y_i(w) 
    ~.
\]

It is difficult to account for $y_i(w)$ accurately according to the actual CCDF, due to the complicated correlation from the OCS.
Instead, we will maintain a lower bound of it, denoted as $\bar{y}_i(w)$, using Lemma~\ref{lem:OCS-gamma}.
For any offline vertex $i$, and any weight-level $w > 0$, consider the sequence of randomized rounds in which $i$ is a candidate.
Further, consider the subset of these rounds in which $i$'s edge weight is at least $w$;
they form a collection of sequences of consecutive rounds in the sense of Lemma~\ref{lem:OCS-gamma}, say, with lengths $k_1, k_2, \dots, k_m$.
By Lemma~\ref{lem:OCS-gamma}, the probability that $i$ is never chosen in these rounds is at most:
\[
    \prod_{\ell=1}^m 2^{-k_\ell} \big(1 - \gamma\big)^{k_\ell-1}
    ~.
\]

Here, we use the fact that the decisions made by the online primal dual algorithm are deterministic, except for the randomness in the OCS. 
In particular, its choices of $i_1$, $i_2$, $i^*$ and the decisions of whether a round is unmatched, randomized, or deterministic are independent of the random bits used by the OCS. 
Hence, we may view the sequence of pairs of candidates sent to the OCS as fixed, and apply Lemma~\ref{lem:OCS-gamma}. 

To maintain $\bar{y}_i(w)$ such that $1 - \bar{y}_i(w)$ is equal to the above upper bound at all time, we update it as follows:

\begin{itemize}
    \item Initially, let $\bar{y}_i(w) = 0$.
    \item Suppose $i$ is matched in a deterministic round, in which its edge weight is at least $w$.
    Then, let $\bar{y}_i(w) = 1$.
    \item Suppose $i$ is matched in a randomized round, in which its edge weight is at least $w$.
    Further, let $w'$ be its edge weight the last time it is chosen in a randomized round;
    $w' = 0$ if it is the first randomized round in which $i$ is involved.
    Then, decrease $1 - \bar{y}_i(w)$ by a $\frac{1}{2} ( 1 - \gamma )$ factor if $w' \ge w$, i.e., if it is not the first step of a consecutive sequence defined by the steps in which $i$ is involved and has edge weight at least $w$; 
    otherwise, decrease $1 - \bar{y}_i(w)$ by half only, to account for the $-1$ in the exponent of $1 - \gamma$.
\end{itemize}

Then, Lemma~\ref{lem:OCS-gamma} ensures that $\bar{y}_i(w)$'s are indeed valid lower bounds of $y_i(w)$'s.
In the rest of the argument, we will use:
\[
    \bar{\alg} = \sum_{i \in L} \int_0^\infty \bar{y}_i(w) dw
    ~,
\]
as a surrogate objective in place of $\alg$ to account for the objective by the algorithm.

\begin{lemma}
    \label{lem:alg-lower-bound}
    We have:
    \[
        \alg \ge \bar{\alg} 
        ~.
    \]
\end{lemma}

In particular, we will obtain the inequality $\alg \ge \dual$ by ensuring that the increment in the surrogate objective $\bar{\alg}$ is at least that in the dual objective $\dual$ in every round.

The next lemma follows by the definition of $\bar{y}_i(w)$'s and $k_i(w)$'s, in particular, by how much $1 - \bar{y}_i(w)$ decreases in a randomized round in which $i$ is a candidate, and its edge weight is at least $w$.

\begin{lemma}
    \label{lem:unmatched-portion-lower-bound}
    For any offline vertex $i$ and any weight-level $w > 0$, we have:
    \[
        1 - \bar{y}_i(w) \ge 2^{-k_i(w)} \big(1 - \gamma\big)^{\max\{k_i(w)-1, 0\}}
        ~.
    \]
\end{lemma}

\begin{proof}
    The value of $1 - \bar{y}_i(w)$ is initially $1$.
    Then, it decreases by $\frac{1}{2}$ in the first of the $k_i(w)$ randomized rounds in which $i$ is a candidate and has edge weight at least $w$.
    Finally, in each of the subsequent $k_i(w) - 1$ rounds, it decreases by at most $\frac{1}{2} \big( 1 - \gamma \big)$.
\end{proof}

As corollaries, we can lower bound the increment in $y_i(w)$.

\begin{lemma}
    \label{lem:primal-increment-deterministic}
    For any offline vertex $i$ and any weight-level $w > 0$, suppose $i$ is chosen as the candidate in a deterministic round, and its edge weight is at least $w$.
    Then, the increment in $\bar{y}_i(w)$ is at least:
    \[
        2^{-k_i(w)} \big(1 - \gamma\big)^{\max\{k_i(w)-1,0\}}
        ~.
    \]
\end{lemma}

\begin{lemma}
    \label{lem:primal-increment-randomized}
    For any offline vertex $i$ and any weight-level $w > 0$, suppose $i$ is chosen as one of the candidates in a randomized round, and its edge weight is at least $w$.
    Then, the increment in $\bar{y}_i(w)$ is at least:
    \[ 
        2^{-k_i(w)-1} \big(1 - \gamma\big)^{\max\{k_i(w)-1, 0\}} 
        ~.
    \]
    Suppose further vertex $i$'s edge weight is at least $w$ the last time it is chosen as a candidate in a randomized round, the increment is at least ($k_i(w) \ge 1$ in this case):
    \[
        2^{-k_i(w)-1} \big(1 - \gamma\big)^{k_i(w)-1} \big(1 + \gamma\big)
        ~.
    \]
\end{lemma}

\begin{proof}
    By definition, $1 - \bar{y}_i(w)$ decreases by a factor of either $\frac{1}{2} (1 - \gamma)$ or $\frac{1}{2}$ in a randomized round, depending on whether vertex $i$'s edge weight is at least $w$ the last time it is chosen in a randomized round.
    Therefore, the increment in $\bar{y}_i(w)$ is either a $\frac{1}{2} (1 + \gamma)$ fraction of $1 - \bar{y}_i(w)$, or a $\frac{1}{2}$ fraction. 
    Putting together with the lower bound of $1 - \bar{y}_i(w)$ in Lemma~\ref{lem:unmatched-portion-lower-bound} proves the lemma.
\end{proof}

\subsubsection{Dual Updates of Offline Vertices: Proof of Eqn.~\eqref{eqn:invariant-alpha}}

\paragraph{Deterministic Rounds.}
The updates of $\alpha_i(w)$'s in a deterministic round is simple.
Fixed any offline vertex $i$.
Suppose it is matched in a deterministic round in which its edge weight is $w_{ij}$.
Then, for any weight-level $w > w_{ij}$, the value of $k_i(w)$ stays the same and thus, $\alpha_i(w)$ stays the same.
For any weight-level $w \le w_{ij}$, on the other hand, the value of $k_i(w)$ becomes $\infty$ by definition.
Therefore, to maintain the invariant in Eqn.~\eqref{eqn:invariant-alpha}, increase $\alpha_i(w)$ by:
\begin{equation}
    \label{eqn:alpha-increment-deterministic}
    \sum_{\ell = k_i(w)}^{\infty} a(\ell)
    ~.
\end{equation}

\paragraph{Randomized Rounds.}
The updates of $\alpha_i(w)$'s in randomized rounds are more subtle.
Fixed any offline vertex $i$.
Suppose it is one of the two candidates in a randomized round in which its edge weight is $w_{ij}$.
Further consider $i$'s edge weight the last time it is chosen in a randomized round, denoted as $w'$; 
let $w' = 0$ if it is the first randomized round in which $i$ is involved.
Then, $w_{ij}$ and $w'$ partition the weight-levels $w > 0$ into up to three subsets, each of which requires a different update rule for $\alpha_i(w)$.
Concretely, increase $\alpha_i(w)$ by:
\begin{equation}
    \label{eqn:alpha-increment-randomized}
    \begin{cases}
        a\big(k_i(w)\big) 
        &
        0 < w \le w_{ij}, w' \textrm{ or } k_i(w) = 0 ~; \\
        a\big(k_i(w)\big) - 2^{-k_i(w)-1} \big(1 - \gamma\big)^{k_i(w)-1} \gamma 
        &
        w' < w \le w_{ij} \textrm{ and } k_i(w) \ge 1 ~; \\
        2^{-k_i(w)-1} \big(1 - \gamma\big)^{k_i(w)-1} \gamma
        &
        w > w_{ij} \textrm{ and } k_i(w) \ge 1 ~.
    \end{cases}
\end{equation}

The first case is the base case of the analysis.
We simply increase $\alpha_i(w)$ by $a\big(k_i(w)\big)$ to maintain the invariant in Eqn.~\eqref{eqn:invariant-alpha}.
For example, in the unweighted version of the problem, all relevant weight-levels $0 < w \le 1$ are always updated according to this base case.

For a weight-level $w$ that falls into the second case (if there is any), the increment in $\alpha_i(w)$ is smaller than that in the base case by $2^{-k_i(w)-1} \big(1 - \gamma\big)^{k_i(w)-1} \gamma$. 
This is the difference between the lower bounds of the increment in $\bar{y}_i(w)$ in Lemma~\ref{lem:primal-increment-randomized} depending on whether $i$'s edge weight is also at least $w$ the last time it is chosen in a randomized round.
Intuitively, since the increase in the surrogate algorithm objective $\bar{\alg}$, due to vertex $i$ and weight-level $w$, is smaller than that in the base case, we subtract this amount from the increment in $\alpha_i(w)$ so that the contribution to $\beta_j$ is unaffected.

How can we still maintain the invariant in Eqn.~\eqref{eqn:invariant-alpha}?
Whenever the second case happens, the same weight-level must fall into the third case in the previous randomized round in which $i$ is involved. 
The same amount is prepaid to $\alpha_i(w)$ back then.
An upper bound of this prepaid amount is subtracted from $\beta_j$, as the second term in Eqn.~\eqref{eqn:beta-increment-randomized}.

\subsubsection{Comparing the Objectives}

This subsection derives a set of sufficient conditions under which the increment in $\bar{\alg}$ is at least that in $\dual$.
Then, the first condition of Lemma~\ref{lem:online-primal-dual-rephrase}, regarding the objectives $\alg$ and $\dual$, follows by:
\[
    \alg \ge \bar{\alg} \ge \dual
    ~.
\]

\paragraph{Deterministic Rounds.}
Suppose $j$ is matched to $i$ in a deterministic round.
Then, by Lemma~\ref{lem:primal-increment-deterministic}, the increment in $\bar{\alg}$ is at least:
\[
    \int_0^{w_{ij}} 2^{-k_i(w)} \big(1-\gamma\big)^{\max\{k_i(w)-1, 0\}} dw
    ~.
\]

In the dual, by Eqn.~\eqref{eqn:alpha-increment-deterministic}, the total increase in $\alpha_i(w)$'s is:
\[
    \int_0^{w_{ij}} \sum_{k = k_i(w)}^\infty a(k) dw
    ~.
\]

On the other hand, by Eqn.~\eqref{eqn:beta-increment-deterministic}, the value of $\beta_j$ is at most:
\[
    \kappa \int_0^{w_{ij}} b\big(k_i(w)\big) dw - \frac{\kappa}{2} \int_{w_{ij}}^{\infty} \sum_{0 \le \ell < k_i(w)} a(\ell) dw 
    \le \kappa \int_0^{w_{ij}} b\big(k_i(w)\big) dw
    ~.
\]

Here, dropping the second term is wlog as we may indeed have $k_i(w) = 0$ for all $w > w_{ij}$ in the worst case.

Putting together, it suffices to ensure that:
\[
    \int_0^{w_{ij}} \sum_{k = k_i(w)}^\infty a(k) dw + \kappa \int_0^{w_{ij}} b\big(k_i(w)\big) dw 
    \le 
    \int_0^{w_{ij}} 2^{-k_i(w)} \big(1-\gamma\big)^{\max\{k_i(w)-1, 0\}} dw
    ~.
\]

Since all terms are integration from $0$ to $w_{ij}$, it suffices to have the inequality without the integration.
That is, it suffices to have that:
\begin{equation}
    \label{eqn:gain-split-deterministic}
    \forall k \ge 0 \quad : \qquad 
    \sum_{\ell = k}^\infty a(\ell) + \kappa \cdot b\big(k\big)
    \le
    2^{-k} \big(1-\gamma\big)^{\max\{k-1, 0\}}
    ~.
\end{equation}

\paragraph{Randomized Rounds.}
Suppose $j$ is matched in a randomized round with $i_1$ and $i_2$ as the two candidates.
We will show that the increment in $\bar{\alg}$ due to $i_1$ is at least the increase in $\alpha_{i_1}(w)$'s plus its contribution to $\beta_j$, i.e., $\Delta_{i_1}^R \beta_j$.
Then, by symmetry, the claim also holds for $i_2$.
Putting together proves that the increase in $\bar{\alg}$ is at least that in $\dual$.

For simplicity of notations, let $w_1$ denote the edge weights of $i_1$ in this round, and further let $w_1'$ denote the edge weight of $i_1$ the last time it is chosen in a randomized round.
Then, $w_1$ and $w_1'$ partition the weight-levels $w > 0$ into up to three subsets that correspond to the three cases of the dual updates of $\alpha_i(w)$ in a randomized round, as in Eqn.~\eqref{eqn:alpha-increment-randomized}.

\paragraph{Case 1: $w \le w_1, w_1'$ or $k_i(w) = 0$.}
Then, by Lemma~\ref{lem:primal-increment-randomized}, the increment in $\bar{\alg}$ due to $i_1$ at weight-level $w$ is at least:
\[
    \begin{cases}
        \frac{1}{2} & k_i(w) = 0 ~; \\
        2^{-k_i(w)-1} \big(1 - \gamma\big)^{k_i(w)-1} \big(1 + \gamma\big) & k_i(w) \ge 1 \textrm{ and } w \le w_1, w_1' ~.
    \end{cases}
\]

By the first case of Eqn.~\eqref{eqn:alpha-increment-randomized}, the increase in $\alpha_i(w)$ is:
\[
    a\big(k_i(w)\big)
    ~.
\]

Finally, the contribution to the first term of $\beta_j$, at weight level $w$, is:
\[
    b\big(k_i(w)\big)
    ~.
\]

Hence, we will ensure that:
\begin{equation}
    \label{eqn:gain-split-randomized}
    \begin{aligned}
        & a(0) + b(0) \le \frac{1}{2} \\[1ex]
        \forall k \ge 1 \quad : \qquad & a(k) + b(k) \le 2^{-k-1} \big(1 - \gamma\big)^{k-1} \big(1 + \gamma\big)
    \end{aligned}
\end{equation}

\paragraph{Case 2: $w_1' < w \le w_1$ and $k_i(w) \ge 1$.}
Then, by Lemma~\ref{lem:primal-increment-randomized}, the increment in $\bar{\alg}$ due to $i_1$ at weight-level $w$ is at least:
\[
    2^{-k_i(w)-1} \big(1 - \gamma\big)^{k_i(w)-1}
    ~.
\]

By the second case of Eqn.~\eqref{eqn:alpha-increment-randomized}, the increase in $\alpha_i(w)$ is:
\[
    a\big(k_i(w)\big) - 2^{-k_i(w)-1} \big(1 - \gamma\big)^{k_i(w)-1} \gamma
    ~.
\]

Finally, the contribution to the first term of $\beta_j$, at weight level $w$, is:
\[
    b\big(k_i(w)\big)
    ~.
\]

Hence, we will ensure that:
\[
    a\big(k_i(w)\big) - 2^{-k_i(w)-1} \big(1 - \gamma\big)^{k_i(w)-1} \gamma + b\big(k_i(w)\big) 
    \le 
    2^{-k_i(w)-1} \big(1 - \gamma\big)^{k_i(w)-1} 
    ~.
\]

Rearranging the second term to the RHS, this becomes the second case of Eqn.~\eqref{eqn:gain-split-randomized}.

\paragraph{Case 3: $w > w_1$ and $k_i(w) \ge 1$.}
The increment in $\bar{\alg}$ due to $i_1$ at weight-level $w$ is zero.

By the last case of Eqn.~\eqref{eqn:alpha-increment-randomized}, the increase in $\alpha_i(w)$ is:
\[
    2^{-k_i(w)-1} \big(1 - \gamma\big)^{k_i(w)-1} \gamma
    ~.
\]

The negative contribution from the second term of $\beta_j$, at weight-level $w$, is:
\[
    \frac{1}{2} \sum_{0 \le k < k_i(w)} a(k)
    ~.
\]

Hence, we will ensure that:
\[
    2^{-k_i(w)-1} \big(1 - \gamma\big)^{k_i(w)-1} \gamma - \frac{1}{2} \sum_{0 \le k < k_i(w)} a(k) \le 0
    ~.
\]

Since the first term is decreasing in $k_i(w)$ while the second term is increasing in it, it suffices to ensure the above inequality for $k_i(w) = 1$, i.e.:
\begin{equation}
    \label{eqn:gain-split-prepaid}
    a(0) \ge \frac{\gamma}{2}
    ~.
\end{equation}

\bigskip

Combining the three cases, we get that the increment in $\bar{\alg}$ in a randomized round is at least the increment in $\dual$.

\subsubsection{Approximate Dual Feasibility}
Finally, we derive a set of sufficient conditions under which dual feasibility holds up to a $\Gamma$ factor, as stated in the second condition of Lemma~\ref{lem:online-primal-dual-rephrase}.
Consider any $i \in L$ and any $j \in R$, and the values of $k_i(w)$'s, $0 < w \le w_{ij}$, at the time when $j$ arrives.

\paragraph{Boundary Condition at the Limit.}
%
First, it may be the case that $k_i(w) = \infty$ for any $0 < w \le w_{ij}$ and $j$ is unmatched.
Then, we may have $\beta_j = 0$ in this round and thus, the contribution from $\alpha_i(w)$'s alone must ensure approximate dual feasibility.
To do so, we will ensure that the value of $\alpha_i(w)$ is at least $\Gamma$ whenever $k_i(w) = \infty$.
By the invariant in Eqn.~\eqref{eqn:invariant-alpha}, it suffices to have:
\begin{equation}
    \label{eqn:feasibility-alpha-infty}
    \sum_{\ell = 0}^\infty a(\ell) \ge \Gamma
    ~.
\end{equation}

Next, consider $5$ different cases depending on how the algorithm matches $j$, in particular, whether the round of $j$ is randomized, deterministic, or unmatched, and whether $i$ is chosen as a candidate.
We first discuss the cases when the round of $j$ is randomized;
then, we will show that the other cases only need weaker conditions.

\paragraph{Case 1: Round of $j$ is a randomized, $i$ is not a candidate.}
The first case is when $j$ is matched in a randomized round but $i$ is not chosen as one of the candidates, i.e., $i \ne i_1, i_2$.
In this case, we have:
\[
    \beta_j = \Delta_{i_1}^R \beta_j + \Delta_{i_2}^R \beta_j
    ~.
\]

Since $i$ is not chosen, both terms on the RHS above is at least $\Delta_i^R \beta_j$.
Therefore, we have:
\[
    \int_0^\infty \alpha_i(w) dw + \beta_j \ge \int_0^\infty \alpha_i(w) dw + 2 \Delta_{i}^R \beta_j 
    ~.
\]

Further, lower bounding $\alpha_i(w)$ by Eqn.~\eqref{eqn:invariant-alpha}, and plugging in the definition of $\Delta_i^R \beta_j$ in Eqn.~\eqref{eqn:beta-increment-randomized}, the above is at least:
\begin{align*}
    & 
    \int_0^\infty \sum_{0 \le \ell < k_i(w)} a(\ell) dw + 2 \bigg( \int_0^{w_{ij}} b\big(k_i(w)\big) dw - \frac{1}{2} \int_{w_{ij}}^\infty \sum_{0 \le \ell < k_i(w)} a(\ell) dw \bigg) \\
    & \qquad
    = \int_0^{w_{ij}} \sum_{0 \le \ell < k_i(w)} a(\ell) dw + 2 \int_0^{w_{ij}} b\big(k_i(w)\big) dw
    ~.
\end{align*}

Therefore, to show that it is at least $\Gamma \cdot w_{ij}$, it suffices to ensure that:
\begin{equation}
    \label{eqn:feasibility-randomized-not-chosen}
    \forall k \ge 0 \quad : \qquad \sum_{0 \le \ell < k} a(\ell) + 2 \cdot b\big(k\big) \ge \Gamma
    ~.
\end{equation}

\paragraph{Case 2: Round of $j$ is randomized, $i$ is a candidate.}
The second case is when $j$ is matched in a randomized round and $i$ is one of the candidates.
Without loss of generality, suppose $i = i_2$, and $i_1$ is the other candidate.
We have:
\[
    \beta_j = \Delta_{i}^R \beta_j + \Delta_{i_1}^R \beta_j
    ~.
\]

We need to lower bound $\Delta_{i_1}^R \beta_j$ in terms of $\Delta_i^R \beta_j$.
Since the algorithm does not choose a deterministic round with $i$ alone, we have:
\[
    \Delta_{i}^R \beta_j + \Delta_{i_1}^R \beta_j \ge \Delta_i^D \beta_j
    ~.
\]

Further, we have $\Delta_i^D \beta_j = \kappa \cdot \Delta_i^R \beta_j$ by Eqn.~\eqref{eqn:beta-increment-deterministic}. 
Putting together, we get that:
\[
    \beta_j \ge \kappa \cdot \Delta_i^R \beta_j
    ~.
\]

Therefore, we have:
\[
    \int_0^\infty \alpha_i(w) dw + \beta_j \ge \int_0^\infty \alpha_i(w) dw + \kappa \cdot \Delta_{i}^R \beta_j 
    ~.
\]

Recall that $k_i(w)$ denote its value at the time when $j$ arrives.
In the round of $j$, the value of $k_i(w)$ increases by $1$ for any weight-level $0 < w \le w_{ij}$, and stays the same for any other weight-level $w > w_{ij}$.
Therefore, lower bounding $\alpha_i(w)$ by Eqn.~\eqref{eqn:invariant-alpha}, the first term on the RHS above is at least:
\[
    \int_0^{w_{ij}} \sum_{0 \le \ell \le k_i(w)} a(\ell) dw + \int_{w_{ij}}^\infty \sum_{0 \le \ell < k_i(w)} a(\ell) dw
    ~.
\]

Further, by the definition of $\Delta_i^R \beta_j$ in Eqn.~\eqref{eqn:beta-increment-randomized}, the second term on the RHS above is:
\[
    \kappa \cdot \bigg( \int_0^{w_{ij}} b\big(k_i(w)\big) dw - \frac{1}{2} \int_{w_{ij}}^\infty \sum_{0 \le \ell < k_i(w)} a(\ell) dw \bigg) 
    ~.
\]

Putting together, and by $\kappa < 2$, we have:
\[
    \int_0^\infty \alpha_i(w) dw + \beta_j \ge \int_0^{w_{ij}} \bigg( \sum_{0 \le \ell \le k_i(w)} a(\ell) + \kappa \cdot b\big(k_i(w)\big) \bigg) dw
    ~
\]

Therefore, to show that the above is at least $w_{ij}$, it suffices to ensure that:
\begin{equation}
    \label{eqn:feasibility-randomized-chosen}
    \forall k \ge 0 \quad : \qquad \sum_{0 \le \ell \le k} a(\ell) + \kappa \cdot b(k) \ge \Gamma
    ~.
\end{equation}

Comparing with Eqn.~\eqref{eqn:feasibility-randomized-not-chosen}, there are two differences.
First, it is important that the above summation includes $\ell = k$, unlike in the previous case.
We can do this because $i$ is one of the two candidates and therefore, $k_i(w)$ increases by $1$ in the round of $j$ for any weight level $w \le w_{ij}$.
On the other hand, the coefficient of the second term is smaller, i.e., $\kappa < 2$.

\paragraph{Case 3: Round of $j$ is deterministic, $i$ is not the candidate.}
The third case is when $j$ is matched to some $i^* \ne i$ in a deterministic round.
Then, we have:
\[
    \beta_j = \Delta_{i^*}^D \beta_j
    ~.
\]

We need to lower bound it in terms of $\Delta_i^R \beta_j$.
Since the algorithm does not choose a randomized round with $i$ and $i^*$ as the candidates, we have:
\[
    \Delta_{i^*}^D \beta_j > \Delta_{i^*}^R \beta_j + \Delta_i^R \beta_j
    ~.
\]

By Eqn.~\eqref{eqn:beta-increment-deterministic}, we have:
\[
    \Delta_{i^*}^R \beta_j = \frac{1}{\kappa} \cdot \Delta_{i^*}^D \beta_j \ge \frac{1}{2} \cdot \Delta_{i^*}^D \beta_j
    ~.
\]

Here, we use the fact that $\Delta_{i^*}^D \beta_j \ge 0$, because $i^*$ is chosen in a deterministic round.
Then, we have:
\[
    \Delta_{i^*}^D \beta_j > 2 \cdot \Delta_i^R \beta_j
    ~.
\]

Therefore, we deduce that:
\[
    \int_0^{w_{ij}} \alpha_i(w) dw + \beta_j > \int_0^{w_{ij}} \alpha_i(w) dw + 2 \cdot \Delta_{i}^R \beta_j 
    ~.
\]

This is identical to the lower bound that we derive in the first case.
Therefore, approximate dual feasibility is guaranteed by Eqn.~\eqref{eqn:feasibility-randomized-not-chosen}.

\paragraph{Case 4: Round of $j$ is deterministic, $i$ is the candidate.}
The fourth case is when $j$ is matched to $i$ in a deterministic round.
Then, for any $0 < w \le w_{ij}$, we have $k_i(w) = \infty$ after this round and thus, we get approximate dual feasibility from the contribution of $\alpha_i(w)$'s alone, due to Eqn.~\eqref{eqn:invariant-alpha} and Eqn.~\eqref{eqn:feasibility-alpha-infty}.

\paragraph{Case 5: Round of $j$ is Unmatched.}
Finally, consider the case when $j$ is left unmatched by the algorithm.
Then, $\beta_j = 0$.
Moreover, it must be that $\Delta_i^D \beta_j < 0$ because the algorithm chooses to leave $j$ unmatched, which further implies $\Delta_i^R \beta_j < 0$ by Eqn.~\eqref{eqn:beta-increment-deterministic}.
Therefore:
\begin{align*}
    \int_0^\infty \alpha_i(w) dw + \beta_j
    &
    = \int_0^\infty \alpha_i(w) dw \\
    &
    > \int_0^\infty \alpha_i(w) dw + 2 \cdot \Delta_i^R \beta_j
    ~.
\end{align*}

Once again, it is the same lower bound that we derive in the first case and therefore, approximately dual feasibility is guaranteed by Eqn.~\eqref{eqn:feasibility-randomized-not-chosen}.

\subsection{Optimizing the Gain Sharing Parameters}
\label{sec:gain-sharing}

To optimize the competitive ratio $\Gamma$ in the above online primal dual analysis, it remains to solve the gain sharing parameters $a(k)$'s and $b(k)$'s from the following LP:
\begin{align*}
    \textrm{maximize} \quad & 
    \Gamma \\
    \textrm{subject to} \quad & \text{Eqn.~\eqref{eqn:gain-split-deterministic}, \eqref{eqn:gain-split-randomized}, \eqref{eqn:gain-split-prepaid}, \eqref{eqn:feasibility-alpha-infty}, \eqref{eqn:feasibility-randomized-not-chosen}, and \eqref{eqn:feasibility-randomized-chosen}}
\end{align*}

We solve a more restricted LP which is finite, by letting $a(k) = b(k) = 0$ for all $k > \kmax$, where $\kmax$ is a sufficiently large integer.
It becomes the following:
\begin{align}
    \textrm{maximize} \quad & 
    \Gamma \label{eqn:ratio-lp} \\[1ex]
    \textrm{subject to} \quad & 
    \sum_{k \le \ell \le \kmax} a(\ell) + \kappa \cdot b(k) \le 2^{-k} \big(1 - \gamma\big)^{\max\{k-1,0\}}  && \forall 0 \le k \le \kmax \notag \\
    &
    a(0) + b(0) \le \frac{1}{2} \notag \\[1ex]
    & 
    a(k) + b(k) \le 2^{-k-1} \big(1 - \gamma\big)^{k-1} \big(1 + \gamma\big) && \forall 1 \le k \le \kmax \notag \\[2ex]
    &
    a(0) \ge \frac{\gamma}{2} \notag \\[1ex]
    &
    \sum_{0 \le \ell \le \kmax} a(\ell) \ge \Gamma \notag \\
    &
    \sum_{0 \le \ell < k} a(\ell) + 2 \cdot b(k) \ge \Gamma && \forall 0 \le k \le \kmax \notag \\
    &
    \sum_{0 \le \ell \le k} a(\ell) + \kappa \cdot b(k) \ge \Gamma && \forall 0 \le k \le \kmax  \notag \\
    & 
    a(k), b(k) \ge 0 && \forall 0 \le k \le \kmax \notag
\end{align}

See Table~\ref{tab:lp-solution} for an approximately optimal solution to the finite LP with $\gamma = \frac{1}{16}$, $\kappa = \frac{3}{2}$ and $\kmax = 7$, which gives $\Gamma > 0.505$.

Finally, we try different values of $\kappa = \frac{\ell}{16}$, $0 \le \ell \le 16$, with $\kmax = 7$.
If $\kappa = 0$ or $\kappa = 1$, $\Gamma = 0.5$;
if $\kappa = \frac{15}{16}$, $\Gamma \approx 0.5026$;
for all other values of $\kappa$, $\Gamma > 0.505$.
We conclude that the analysis is robust to the choice of $\kappa$, so long as it is neither too close to $1$ nor to $2$.

\clearpage

\begin{table}[t]
    \centering
    \renewcommand{\arraystretch}{1.2}
    \begin{tabular}{c|cc}
        \hline
        $k$ & $a(k)$ & $b(k)$ \\
        \hline
        $0$ & $0.24749974$ & $0.25250026$ \\
        $1$ & $0.13687460$ & $0.12875040$ \\
        $2$ & $0.06419862$ & $0.06031309$ \\
        $3$ & $0.03015109$ & $0.02821378$ \\
        $4$ & $0.01422029$ & $0.01313824$ \\
        $5$ & $0.00679622$ & $0.00602809$ \\
        $6$ & $0.00338141$ & $0.00262999$ \\
        $7$ & $0.00187856$ & $0.00093928$ \\
        \hline
    \end{tabular}
    \caption{An approximately optimal solution to the finite LP with $\kappa = \frac{3}{2}$, $\kmax = 7$, and $\gamma = \frac{1}{16}$, rounded to the $8$-th digit after the decimal point, whose $\Gamma \approx 0.50500053$.}
    \label{tab:lp-solution}
\end{table}

\section{Improved Online Correlated Selection}
\label{sec:improved-OCS}

This section gives a $\frac{1}{3\sqrt{3}}$-OCS, which leads to the competitive ratio stated in Theorem~\ref{thm:main}.
We present the algorithm in Section~\ref{sec:improved-OCS-algorithm}, and the analysis in Section~\ref{sec:improved-OCS-analysis}.
Finally, we show an approximately optimal solution to the LP in the previous section with the improved $\gamma = \frac{1}{3\sqrt{3}}$ in Section~\ref{sec:improved-OCS-LP-solution}.

\subsection{Algorithm}
\label{sec:improved-OCS-algorithm}

Given a fixed sequence of pairs of candidates, define the \emph{ex-ante dependence graph} to be a directed graph $G^{\exante} = (V, E^{\exante})$ as follows.
To make a distinction with the vertices and edges in the online matching problem, we will refer to the elements in $V$ as \emph{nodes}, and those in $E^{\exante}$ as \emph{arcs}.
Let there be a node in $V$ for each pair of candidates.
In the context of the edge-weighted online bipartite matching algorithm in this paper, each pair corresponds to an online vertex $j \in R$ matched in a randomized round, and the candidates are offline vertices.
Therefore, we let $j$ be the indices of pairs and $V \subseteq R$, and let $i$ be the indices of candidates.
The \emph{ex-ante} graph may be a multi-graph because, for example, there may be two consecutive pairs $j$ and $j+1$ that have the same $i_1$ and $i_2$, which lead to two parallel arcs from $j$ to $j+1$.
To this end, we will write an arc as $(j, j')_i$ where $i$ is the common candidate in rounds $j$ and $j'$ which leads to the arc.
See Figure~\ref{fig:exante} for an example.

The main ingredient of the improved OCS is an online algorithm that constructs the \emph{ex-post dependence graph}, denoted as $G^{\expost} = (V, E^{\expost})$, which is a subgraph of $G^{\exante}$.
To realize the \emph{ex-post} dependence graph, each node independently and randomly chooses one of the two types: \emph{sender} and \emph{receiver}.
The probability of being a sender, denoted as $p$, will be optimized later.

A sender node uses a fresh random bit to choose one of the two candidates, and is going to send the bit along the out-arcs in the \emph{ex-post} graph to introduce negative correlation.

A receiver node, on the other hand, seeks to read the random bit it receives and makes the opposite decision.
First, it checks both in-arcs in the \emph{ex-ante} graph to see if any of its in-neighbors is a sender.
If so, the receiver node reads the random bit realized in the round of such an in-neighbor;
if both in-neighbors are senders, choose one uniformly at random.
The corresponding arc is added to the \emph{ex-post} dependence graph.
Suppose a receiver node $j$ receives the random bit of a sender node $j'$ sent along arc $(j', j)_i$.
Then, the improved OCS chooses $i$ in round $j$ if it is not chosen in round $j'$, and vice versa.

\begin{figure}
    \centering
    \begin{subfigure}{\textwidth}
        \centering
        \includegraphics[width=.9\textwidth]{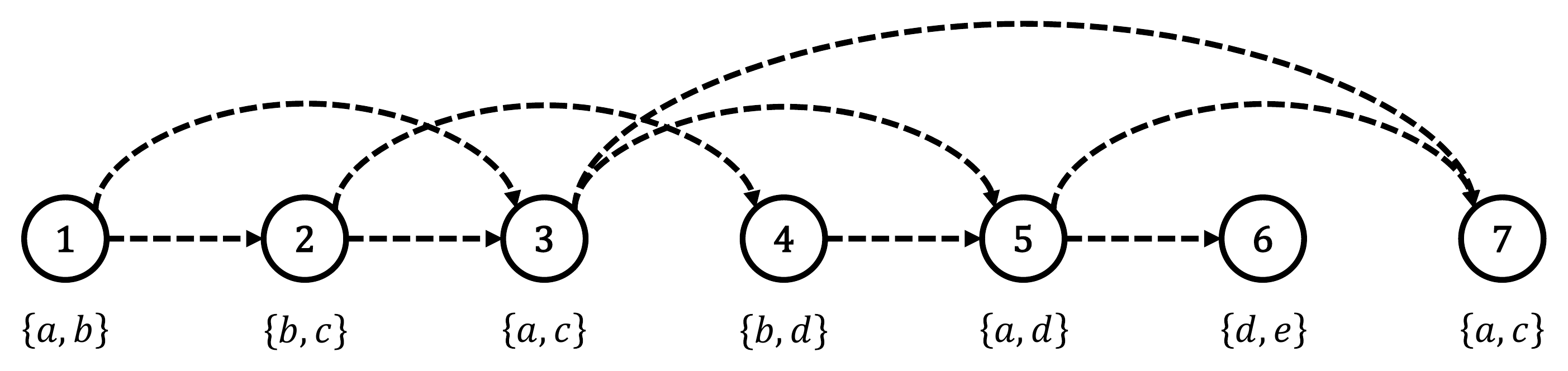}
        \caption{\emph{Ex-ante} dependence graph}
        \label{fig:exante}
    \end{subfigure}

    \begin{subfigure}{\textwidth}
        \centering
        \includegraphics[width=.9\textwidth]{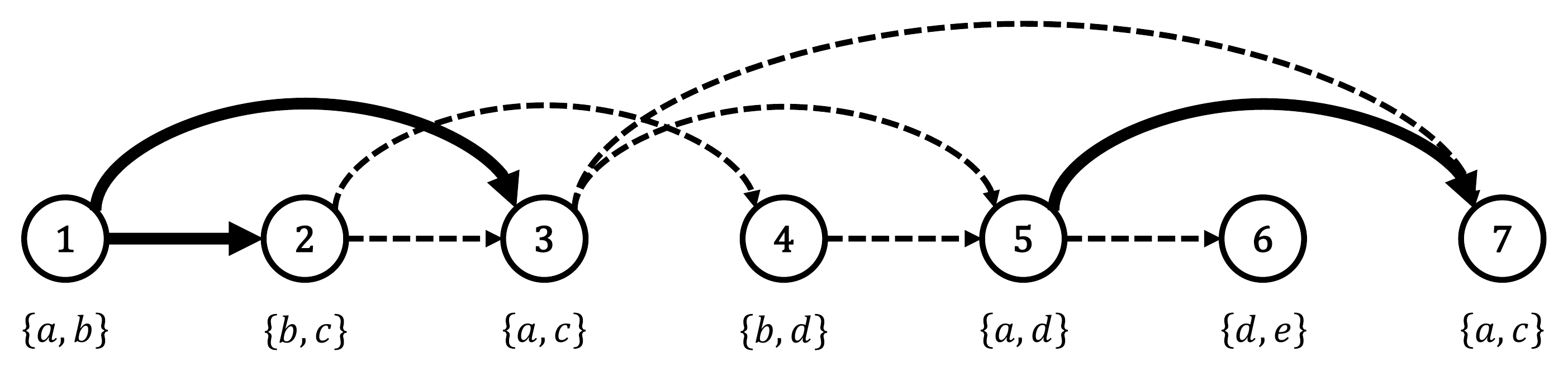}
        \caption{\emph{Ex-post} dependence graph (bold and solid edges). In this example, node $1$ draws $\sender$; nodes $2$ and $3$ draw $\receiver$; node $4$ draws $\sender$, but node $5$ also draws $\sender$ and thus, arc $(4, 5)_d$ is not realized; node $6$ draws $\sender$ and thus, arc $(5, 6)_d$ is not realized; node $7$ draws $\receiver$.}
        \label{fig:expost}
    \end{subfigure}
    \caption{An example of the \emph{ex-ante} and \emph{ex-post} dependence graphs for a sequence of $7$ pairs and $5$ different candidates $a$, $b$, $c$, $d$, and $e$.}
    \label{fig:dependence-graphs}
\end{figure}

\begin{algorithm}[t!]
    \caption{Improved Online Correlated Selection}
    \label{alg:improved-OCS}
    \begin{algorithmic}
        \medskip
        \STATEx \textbf{Parameter:}
        \begin{itemize}
            \item $p$ : probability that a node is a $\sender$.
        \end{itemize}
        \STATEx \textbf{State variables:}
        \begin{itemize}
            \item $G^{\text{\em ex-ante}} = (V, E^{\exante})$ : \emph{ex-ante} dependence graph; initially, $V = E^{\exante} = \emptyset$.
            \item $G^{\text{\em ex-post}} = (V, E^{\expost})$ : \emph{ex-post} dependence graph; initially, $V = E^{\expost} = \emptyset$.
            \item $\tau_j \in \big\{ \sender, \receiver \big\}$ for any $j \in V$.
        \end{itemize}
        \STATEx \textbf{On receiving $2$ candidate offline vertices $i_1$ and $i_2$ (for an online vertex $j \in R$):}
        \begin{enumerate}
            \item Add $j$ to $V$.
            \item For $k \in \{1, 2\}$, let $j_k$ be the last randomized round in which $i_k$ is a candidate;
                add an arc $(j_k, j)_{i_k}$ to $E^{\exante}$, i.e., the \emph{ex-ante} dependence graph.
            \item Let $\tau_j = \sender$ with probability $p$, and $\tau_j = \receiver$ with probability $1 - p$.
                \label{step:improved-OCS-type}
            \item If $\tau_j = \sender$, let $i^* = i_1$ or $i_2$, each with probability $\frac{1}{2}$.
                \label{step:improved-OCS-sender-choice}
            \item If $\tau_j = \receiver$:
            \begin{enumerate}
                \item If there is $j_m$, $m \in \{ 1, 2 \}$, so that $\tau_{j_m} = \sender$ (break ties randomly):
                    \label{step:improved-OCS-tie}
                \begin{enumerate}
                    \item Add an arc $(j_m, j)_{i_m}$ to $E^{\expost}$, i.e., the \emph{ex-post} dependence graph.
                    \item Let $i^* = i_m$ if $i_m$ is not chosen in round $j_m$, let $i^* = i_{-m}$ otherwise.
                \end{enumerate}
                \item Otherwise, let $i^* = i_1$ or $i_2$, each with probability $\frac{1}{2}$.
                    \label{step:improved-OCS-receiver-choice}
            \end{enumerate}
            \item Return $i^*$.
        \end{enumerate}
        \smallskip
    \end{algorithmic}
\end{algorithm}

See Algorithm~\ref{alg:improved-OCS} for a formal definition of the improved OCS.

\paragraph{Comparison with the Original OCS.}
A brief comparison with the original OCS defined in Algorithm~\ref{alg:OCS} is due.
The original one can also be interpreted in the language of the \emph{ex-ante} and \emph{ex-post} dependence graphs, as in its analysis by \citet{HuangT/arXiv/2019}.
It constructs the \emph{ex-post} dependence graph as follows:
on the arrival of each node $j$ with candidates $i_1$ and $i_2$, it picks one of its up to four incident arcs, i.e., the two in-arcs and two out-arcs due to $i_1$ and $i_2$, each with probability $\frac{1}{4}$;
it adds an arc $(j', j)_i$ to the \emph{ex-post} dependence graph if both nodes picks this arc, i.e., if $j'$ picks the out-arc due to $i$, and $j$ picks the in-arc due to $i$.
Then, the random bits are sent along the arcs in the \emph{ex-post} dependence graph the same way as in the improved OCS.

The main difference is that the original version constructs the \emph{ex-post} dependence graph to be a matching, while the improved version further allows v-structures in which a sender has two out-arcs to two receivers.
It is based on an observation that such v-structures do not introduce undesirable positive correlation in the matching decisions of the same candidate in different rounds.
We formalize it as the next lemma.

\begin{lemma}
    \label{lem:improved-OCS-CC-correlation}
    Suppose two rounds $j_1$ and $j_2$ have a common candidate $i$.
    Then:
    \begin{enumerate}
        \item If $j_1$ and $j_2$ belong to the same connected component in the \emph{ex-post} dependence graph, $i$'s matching decisions in $j_1$ and $j_2$ are the opposite.
        \item Otherwise, $i$'s matching decisions in $j_1$ and $j_2$ are independent.
    \end{enumerate}
\end{lemma}

\begin{proof}
    \emph{(Part 1)~}
    There are two possibilities.
    First, there may be an arc between $j_1$ and $j_2$ in the \emph{ex-post} dependence graph. 
    Then, regardless of whether the arc is due to $i$ or another common candidate, the matching decisions are the opposite by definition.
    See candidate $a$ and nodes $1$ and $3$ in Figure~\ref{fig:expost} for an example.

    Second, $j_1$ and $j_2$ may be the receivers of the same sender $j_s$ in a v-structure.
    That is, $(j_s, j_1)$ and $(j_s, j_2)$ are both in the \emph{ex-post} graph; 
    suppose $i_1$ and $i_2$ are the candidates in $j_s$, $i_1$ and $i$ are the candidates in $j_1$, and $i_2$ and $i$ are the candidates in $j_2$.
    We still get that $i$'s matching decisions in $j_1$ and $j_2$ are perfectly negatively correlated, due to a sequence of perfect negative correlation: $i$ and $i_1$'s matching decisions in $j_1$, $i_1$'s matching decisions in $j_s$ and $j_1$, $i_1$ and $i_2$'s matching decisions in $j_s$, $i_2$'s matching decisions in $j_s$ and $j_2$, and $i_2$ and $i$'s matching decisions in $j_2$.
    See candidate $c$ and nodes $2$ and $3$ in Figure~\ref{fig:expost} for an example.

    \emph{(Part 2)~} 
    This follows by the definition of the improved OCS.
    Even if one or both of $j_1$ and $j_2$ are receiver nodes, from $i$'s viewpoint the matching decisions are still effectively fresh random bits independent of each other; 
    the bits are merely drawn earlier by some sender nodes and sent along the arcs.
    See candidate $c$ and nodes $3$ and $7$ in Figure~\ref{fig:expost} for an example.
    Even though $c$'s matching decisions are based on random bits drawn in nodes $1$ and $5$ respectively, they are independent decisions from $c$'s viewpoint.
\end{proof}

The main property of the improved OCS is the next lemma, whose proof is deferred to to next subsection.

\begin{lemma}
    \label{lem:improved-OCS}
    For any fixed sequence of pairs of candidates, any fixed candidate $i$, and any disjoint sequences of consecutive rounds of lengths $k_1, k_2, \dots, k_m \ge 1$ in which $i$ is involved, the improved OCS in Algorithm~\ref{alg:improved-OCS} ensures that $i$ is chosen in at least one of the rounds with probability at least:
    \[
        1 - \prod_{\ell = 1}^m 2^{-k_\ell} \cdot g_{k_\ell} 
        ~,
    \]
    where $g_k$'s are defined recursively as follows:
    \begin{equation}
        \label{eqn:improved-OCS-recurrence}
        g_k = \begin{cases}
            1 & k = 0, 1 ~; \\
            g_{k-1} - p \big( 1-p \big) \big( 1 - \tfrac{p}{2} \big) \cdot g_{k-2} & k \ge 2 ~.
        \end{cases}
    \end{equation}
\end{lemma}

\begin{lemma}
    The improved OCS in Algorithm~\ref{alg:improved-OCS} is a $p \big( 1-p \big) \big( 1 - \tfrac{p}{2} \big)$-OCS.
\end{lemma}

\begin{proof}
    By the recurrence in Eqn.~\eqref{eqn:improved-OCS-recurrence}, we have $g_k \le g_{k-1}$, which further implies for $k \ge 2$:
    \begin{align*}
        g_k & = g_{k-1} - p \big( 1-p \big) \big( 1 - \tfrac{p}{2} \big) \cdot g_{k-2} \\
        & \le \big( 1 - p \big( 1-p \big) \big( 1 - \tfrac{p}{2} \big) \big) \cdot g_{k-1} ~.
    \end{align*}
    
    Putting together with $g_1 = 1$ and Lemma~\ref{lem:improved-OCS} proves the lemma.
\end{proof}

Further, let $p = 1 - \frac{1}{\sqrt{3}}$ to maximize $p \big( 1-p \big) \big( 1 - \tfrac{p}{2} \big) = \frac{1}{3\sqrt{3}}$.

\begin{lemma}
    \label{lem:improved-OCS-gamma}
    The improved OCS in Algorithm~\ref{alg:improved-OCS} is a $\frac{1}{3\sqrt{3}}$-OCS with $p = 1 - \frac{1}{\sqrt{3}}$.
\end{lemma}

\subsection{Analysis: Proof of Lemma~\ref{lem:improved-OCS}}
\label{sec:improved-OCS-analysis}

Let $j^\ell_1 < j^\ell_2 < \dots < j^\ell_{k_\ell}$, $1 \le \ell \le m$, be the sequences of consecutive rounds in which $i$ is a candidate.
The algorithm uses three kinds of independent random bits:
\begin{itemize}
    \item The random type of each round, realized in Step~\ref{step:improved-OCS-type}; 
        recall that it is denoted as $\tau_j$, which equals either $\sender$ or $\receiver$.
    \item The tie-breaking in Step~\ref{step:improved-OCS-tie};
        we will denote it as $m_j$.
    \item The random choice between $i_1$ and $i_2$ in Step~\ref{step:improved-OCS-sender-choice} or in Step~\ref{step:improved-OCS-receiver-choice} (each round reaches at most one of these two steps); 
        we will denote it as $c_j$. 
\end{itemize}

For convenience of the argument, we will assume that $m_j$'s and $c_j$'s are realized in all rounds, regardless of whether the algorithm reaches the corresponding steps;
they are merely left unused sometimes.

Further, the realization of the \emph{ex-post} dependence graph only relies on the first two kinds of random bits.
Let $F$ denote the event that no two nodes in the sequences belong to the same connected component in the \emph{ex-post} dependence graph, i.e.:
\[
    F = \bigg\{ (\tau, m) : \textrm{$j^\ell_s$, $1 \le \ell \le m$, $1 \le s \le k_\ell$, are in different connected components in $G^{\expost}$} \bigg\}
    ~.
\]

\begin{lemma}
    The probability that $i$ is never chosen is:
    \[
        2^{- \sum_{1 \le \ell \le m} k_\ell} \cdot \Pr \big[ F \big]
        ~.
    \]
\end{lemma}

\begin{proof}
    Consider the realization of the randomness in two stages.
    First, realize the \emph{ex-post} dependence graph by drawing the first two kinds of random bits.
    Then, draw the third kinds of random bits, in particular, for every node with no in-arcs in the \emph{ex-post} dependence graph, i.e., senders, and receivers whose in-neighbors are not senders.

    If there are two nodes in the sequences that belong to the same connected component in the \emph{ex-post} dependence graph, $i$ must be chosen in exactly one of the two rounds by the first part of Lemma~\ref{lem:improved-OCS-CC-correlation}.
    Otherwise, the second part of Lemma~\ref{lem:improved-OCS-CC-correlation} states that $i$'s matching decisions in these rounds are independent. 

    Hence, the probability that $i$ is never chosen is equal to the product of (1) the probability that the nodes in the sequences belong to $\sum_{\ell = 1}^m k_\ell$ different connected components in the \emph{ex-post} dependence graph, and (2) all $\sum_{\ell = 1}^m k_\ell$ independent random bits are against $i$.
    The former is $\Pr \big[ F \big]$.
    The latter is equal to $2^{-\sum_{\ell = 1}^m k_\ell}$.
\end{proof}

Therefore, it remains to show that:
\begin{equation}
    \label{eqn:improved-OCS-cc-bound}
    \Pr \big[ F \big] \le \prod_{\ell = 1}^m g_{k_\ell} 
    ~.
\end{equation}

What are the arcs of concern in event $F$?
By the assumption that these are sequences of consecutive rounds in which $i$ is a candidate, the arcs of the form $(j^\ell_s, j^\ell_{s+1})$ always exist in the \emph{ex-ante} dependence graph.
To characterize whether some of these arcs are realized in the \emph{ex-post} graphs, we need to further consider another set of arcs as follows.

For any $1 \le \ell \le m$, consider the in-arcs of nodes $j^\ell_2 < \dots < j^\ell_{k_\ell}$ in the \emph{ex-ante} dependence graph, other than $(j^\ell_1, j^\ell_2)_i, \dots, (j^\ell_{k_\ell-1}, j^\ell_{k_\ell})_i$.
Let them be $(\hat{j}^\ell_s, j^\ell_s)$, $1 < s \le k_\ell$.
We omit the subscript that denotes the common candidate in the two nodes, with the understanding that they are due to the candidate other than $i$ in the round of $j^\ell_s$.

Then, the realization of $\tau_j$'s for nodes $j^\ell_s$'s and $\hat{j}^\ell_s$'s, and the realization of $m_j$'s for nodes $j^\ell_s$'s determine whether the arcs of the form $(j^\ell_s, j^\ell_{s+1})_i$ are realized.
Concretely, an arc $(j^\ell_s, j^\ell_{s+1})_i$ is realized in the \emph{ex-post} graph if and only if:
\begin{enumerate}
    \item Node $j^\ell_s$ is a sender (realization of $\tau_{j^\ell_s}$).
    \item Node $j^\ell_{s+1}$ is a receiver (realization of $\tau_{j^\ell_{s+1}}$).
    \item Either node $\hat{j}^\ell_{s+1}$ is a receiver (realization of $\tau_{\hat{j}^\ell_{s+1}}$), or the tie-breaking is in favor of $j^\ell_s$ (realization of $m_{j^\ell_{s+1}}$).
\end{enumerate}

\paragraph{Relaxed Event.}
In the rest of the argument, we will consider a slightly relaxed event $G$, which only precludes nodes in the sequences from being connected in the \emph{ex-post} dependence graph in the following forms:
\begin{itemize}
    \item An arc from $j^\ell_s$ to $j^{\ell'}_{s'}$, \emph{and $s' > 1$}.
    \item A v-structure with $j^\ell_s$ and $j^{\ell'}_{s'}$ as the receivers of the same sender, \emph{and $s, s' > 1$}.
\end{itemize}

By definition, we have $F \subseteq G$ and therefore:
\begin{equation}
    \label{eqn:improved-OCS-relaxed-event}
    \Pr \big[ F \big] \le \Pr \big[ G \big]
    ~.
\end{equation}

\paragraph{Regular Case.}
The special case when all $\hat{j}^\ell_s$'s exist, and all $j^\ell_s$'s and $\hat{j}^\ell_s$'s are distinct nodes is of particular importance.
It is relatively easy to analyze because in this special case it suffices to consider arcs of the form $(j^\ell_s, j^\ell_{s+1})_i$, and (2) different sequences may be analyze separately as the relevant nodes and the corresponding random bits are independent.
We will refer to it as the \emph{regular} case.
We will analyze the regular case in Lemma~\ref{lem:improved-OCS-regular}, and will show in Lemma~\ref{lem:improved-OCS-irregular} that this is the worst-case scenario that maximizes $\Pr \big[ G \big]$.

\begin{lemma}
    \label{lem:improved-OCS-regular}
    In the regular case, i.e., when all $\hat{j}^\ell_s$'s exist, and all $j^\ell_s$'s and $\hat{j}^\ell_s$'s are distinct, the probability of event $G$ is: 
    \[
        \Pr \big[ G \big] = \prod_{\ell = 1}^m g_{k_\ell} 
        ~.
    \]
\end{lemma}

\begin{proof}
    We start by formalizing the aforementioned implications of the assumption that all $j^\ell_s$'s and $\hat{j}^\ell_s$'s are distinct.
    First, two nodes in the sequences are connected \emph{in a form precluded by event $G$} if and only if they are consecutive nodes in the same sequence, e.g., $j^\ell_{s-1}$ and $j^\ell_s$, and arc $(j^\ell_{s-1}, j^\ell_s)_i$ is realized.
    A node $j^\ell_s$ with $s > 1$ cannot be the receiver of a sender other than $j^\ell_{s-1}$ in the sequences because $\hat{j}^\ell_s$'s are not in the sequences by the assumption. 
    Two nodes $j^\ell_s$ and $j^{\ell'}_{s'}$ with $s, s' > 1$ cannot be the two receivers of the same sender in a v-structure because $\hat{j}^\ell_s$ and $\hat{j}^{\ell'}_{s'}$ are distinct by the assumption.
    Second, the realization of these arcs in different sequences are independent. 
    The realization of arcs of the form $(j^\ell_{s-1}, j^\ell_s)_i$, for any fixed $\ell$, depends only on the realization of $\tau_j$'s for nodes $j^\ell_s$'s and $\hat{j}^\ell_s$'s, and the realization of $m_j$'s for nodes $j^\ell_s$'s, i.e., nodes with the same superscript $\ell$ (and because of the assumption in the lemma that all $j^\ell_s$'s and $\hat{j}^\ell_s$'s are distinct).

    Next, we focus on a fixed sequence $\ell$ and analyze the probability that no arc of the form $(j^\ell_s, j^\ell_{s+1})_i$, $1 \le s < k_\ell$, is realized.
    For simplicity of notations, we will omit the superscripts and subscripts $\ell$ and write $j_1 < j_2 < \dots < j_k$ and $\hat{j}_2, \hat{j}_3, \dots, \hat{j}_k$.
    Let $G_k$ denote this event, and $g_k$ be its probability.
    Trivially, we have $g_0 = g_1 = 1$.
    It remains to show that $g_k$ follows the recurrence in Eqn.~\eqref{eqn:improved-OCS-recurrence}.
    
    We will do so by further considering an auxiliary subevent $A_k$, which requires not only $G_k$ to happen, but also $j_k$ to be a receiver.
    Let $a_k$ denote the probability of this event for any $k \ge 1$.

    If $j_k$ is a sender, which happens with probability $p$, arc $(j_{k-1}, j_k)_i$ would not be realized regardless of the randomness of the other nodes.
    Therefore, conditioned on $j_k$ being a sender, event $G_k$ reduces to event $G_{k-1}$. 

    If $j_k$ is a receiver, but $\hat{j}_k$ is a sender, and would win the potential tie-breaking, which happens with probability $\frac{p(1-p)}{2}$, we still have that arc $(j_{k-1}, j_k)_i$ would not be realized regardless of the randomness of the other nodes.
    Therefore, conditioned on being in this case, events $G_k$ and $A_k$ reduce to event $G_{k-1}$. 

    Otherwise, which happens with probability $(1 - p)(1 - \frac{p}{2})$, $j_{k-1}$ must not be a sender, or else arc $(j_{k-1}, j_k)_i$ would be realized.
    Therefore, conditioned on being in this case, events $G_k$ and $A_k$ reduces to event $A_{k-1}$. 

    Putting together, we have:
    \begin{align*}
        g_k 
        & 
        = \big( p + \tfrac{p(1-p)}{2} \big) g_{k-1} + \big( 1-p \big) \big(1 - \tfrac{p}{2} \big) a_{k-1} 
        ~; \\[1ex]
        a_k
        &
        = \tfrac{p(1-p)}{2} g_{k-1} + \big( 1-p \big) \big(1 - \tfrac{p}{2} \big) a_{k-1} 
        ~.
    \end{align*}

    Eliminating $a_k$'s by combining the two equations, we get that:
    \[
        g_k = g_{k-1} - p \big( 1-p \big) \big( 1 - \tfrac{p}{2} \big) \cdot g_{k-2} 
        ~.
    \]

    This is indeed the recurrence in Eqn.~\eqref{eqn:improved-OCS-recurrence}.
\end{proof}

\begin{lemma}
    \label{lem:improved-OCS-irregular}
    The probability of event $G$ is maximized in the regular case.
\end{lemma}

\begin{proof}
    Here are the possible violations of the conditions of the regular case: 
    \begin{enumerate}
        \item Some arc $(\hat{j}^\ell_s, j^\ell_s)$ may not exist.
        \item There may be $\ell, \ell', s, s'$ such that $\hat{j}^\ell_s = j^{\ell'}_{s'}$, i.e., the candidate other than $i$ in round $j^\ell_s$ is also a candidate in round $j^{\ell'}_{s'}$, and in no other rounds in between.
        \item  There may be $\ell, \ell', s, s'$ such that $\hat{j}^\ell_s = \hat{j}^{\ell'}_{s'}$, in which case $j^\ell_s$ and $j^{\ell'}_{s'}$ may belong to the same connected component due to a v-structure.
    \end{enumerate}

    We will refer to them as irregular cases of type $1$, $2$, and $3$ respectively.
    
    To compare the probability of event $G$ in a general, potentially irregular case with its probability in the regular case, we consider the realization of the relevant random bits in two steps.
    First, realize the types $\tau_j$ of the nodes in the sequences, i.e., $j^\ell_s$, $1 \le \ell \le m$, $1 \le s \le k_\ell$.
    Then, realize the types of the other nodes as well as the tie-breaking bits $m_j$ for the nodes in the sequences.
    We will show that conditioned on any realization of the random bits in the first step, the probability of $G$ is maximized in the regular case, over the randomness in the second step alone.

    After realizing the random bits in the first step, we may identify potential pairs of consecutive nodes in a sequence, e.g., $j^\ell_{s-1}$ and $j^\ell_s$, such that $j^\ell_{s-1}$ is a sender while $j^\ell_s$ is a receiver. 
    For each pair of such nodes, we would have an arc between them in the \emph{ex-post} dependence graph which stops event $G$ from happening, unless the randomness in the second step is such that $\hat{j}^\ell_s$ is a sender and further wins the tie-breaking, i.e., arc $(\hat{j}^\ell_s, j^\ell_s)$ is realized in the \emph{ex-post} dependence graph.

    However, having arc $(\hat{j}^\ell_s, j^\ell_s)$ realized in the \emph{ex-post} dependence graph does not necessarily help event $G$ happen, because the out-arcs of node $\hat{j}^\ell_s$ may lead to a connected component that is precluded by event $G$.

    To this end, define an indicator $r^\ell_s$ for any $1 \le \ell \le m$ and any $1 < s \le k_\ell$ such that such that $j^\ell_{s-1}$ is a sender while $j^\ell_s$ is a receiver. 
    Let $r^\ell_s$ equal $1$ if the type of node $\hat{j}^\ell_{s}$ and the tie-breaking bit of node $j^\ell_{s}$ are such that (1) arc $(\hat{j}^\ell_s, j^\ell_s)$ is realized in the \emph{ex-post} dependence graph, and (2) the realized out-arcs of node $\hat{j}^\ell_s$ do not lead to a connected component precluded by event $G$.
    Then, event $G$ can be rewritten as:
    \[
        G = \bigg\{ (\tau, m) : \textrm{there is no $\ell$ and $s$ such that $j^\ell_{s-1}$ is a sender, $j^\ell_s$ is a receiver, and $r^\ell_s = 0$} \bigg\} ~.
    \]

    In particular, event $G$ is monotone in the realization of $r^\ell_s$'s.
    It suffices to show that the joint distribution of $r^\ell_s$'s in the regular case stochastically dominates its counterpart in the irregular cases.
    \begin{itemize}
        \item \emph{Regular case:~} 
            Node $\hat{j}^\ell_s$ is distinct from the other $\hat{j}^{\ell'}_{s'}$'s, and from the nodes $j^{\ell'}_{s'}$ in the sequences.
            Then, $r^\ell_s = 1$ if and only if node $\hat{j}^\ell_s$ is a sender and the tie-breaking bit of node $j^\ell_s$ is in favor of  $\hat{j}^\ell_s$. 
            The former happens with probability $p$ while the latter happens with probability $\frac{1}{2}$.
            Hence, we have $r^\ell_s = 1$ with probabilty $\frac{p}{2}$.
        \item \emph{Type-$1$ irregular case:~} 
            Node $\hat{j}^\ell_s$ does not exists.
            Then, $r^\ell_s = 0$.
        \item \emph{Type-$2$ irregular case:~} 
            Node $\hat{j}^\ell_s$ is a node in the sequences. 
            Then, we still have $r^\ell_s = 0$ because the whenever arc $(\hat{j}^\ell_s, j^\ell_s)$ is realized, the arc itself is precluded by event $G$.
        \item \emph{Type-$3$ irregular case:~} 
            Node $\hat{j}^\ell_s = \hat{j}^{\ell'}_{s'}$ for some $(\ell', s') \ne (\ell, s)$.
            If $(\ell', s')$ itself does not satisfy that $j^{\ell'}_{s'-1}$ is a sender and $j^{\ell'}_{s'}$ is a receiver, it is identical to the regular case, i.e., $r^\ell_s = 1$ with probability $\frac{p}{2}$.
            
            Otherwise, we shall define $r^\ell_s$ and $r^{\ell'}_{s'}$ jointly.
            If node $\hat{j}^\ell_s = \hat{j}^{\ell'}_{s'}$ is a receiver, $r^\ell_s = r^{\ell'}_{s'} = 0$.
            If it is a sender, but loses the tie-breaking in both nodes $j^\ell_s$ and $j^{\ell'}_{s'}$, $r^\ell_s = r^{\ell'}_{s'} = 0$.
            If it is a sender, and wins the tie-breaking in both nodes $j^\ell_s$ and $j^{\ell'}_{s'}$, we still have $r^\ell_s = r^{\ell'}_{s'} = 0$ because it is a v-structure precluded by event $G$.
            Finally, if it is a sender and wins the tie-breaking in exactly one of nodes, say, $j^\ell_s$, let $r^\ell_s = 1$ and $r^{\ell'}_{s'}$, and vice versa.
            In sum:
            \[
                \big( r^\ell_s, r^{\ell'}_{s'} \big) = \begin{cases}
                    (0, 0) & \textrm{with probability $1 - \frac{1}{2} p$;} \\
                    (0, 1) & \textrm{with probabiltty $\frac{1}{4} p$;} \\
                    (1, 0) & \textrm{with probabiltty $\frac{1}{4} p$.}
                \end{cases}
            \]
            Importantly, it is stochastically dominated by the joint distribution below when both $r^\ell_s$ and $r^{\ell'}_{s'}$ were regular cases:
            \[
                \big( r^\ell_s, r^{\ell'}_{s'} \big) = \begin{cases}
                    (0, 0) & \textrm{with probability $(1-\frac{p}{2})^2$;} \\
                    (0, 1) & \textrm{with probabiltty $\frac{p}{2}(1-\frac{p}{2})$;} \\
                    (1, 0) & \textrm{with probabiltty $\frac{p}{2}(1-\frac{p}{2})$;} \\
                    (1, 1) & \textrm{with probability $\frac{p^2}{4}$.}
                \end{cases}
            \]
    \end{itemize}
    
    In sum, all irregular cases are stochastically dominated by the regular case.
\end{proof}

Finally, combining Eqn.~\eqref{eqn:improved-OCS-relaxed-event}, Lemma~\ref{lem:improved-OCS-regular}, and Lemma~\ref{lem:improved-OCS-irregular} gives Eqn.~\eqref{eqn:improved-OCS-cc-bound}.
This finishes the proof of Lemma~\ref{lem:improved-OCS}.

\subsection{Application: Proof of Theorem~\ref{thm:main}}
\label{sec:improved-OCS-LP-solution}

We solve the finite LP in Eqn.~\eqref{eqn:ratio-lp} with $\gamma = \frac{1}{3\sqrt{3}}$ as stated in Lemma~\ref{lem:improved-OCS-gamma}, and with $\kmax = 7$ and $\kappa = \frac{3}{2}$.
The competitive ratio is $\Gamma \approx 0.51461$.
See Table~\ref{tab:lp-solution-improved} for an approximately optimal solution.

\begin{table}[t]
    \centering
    \renewcommand{\arraystretch}{1.2}
    
    \begin{tabular}{c|cc}
        \hline
        $k$ & $a(k)$        & $b(k)$        \\
        \hline
        $0$ & $0.24269440$  & $0.25730560$  \\
        $1$ & $0.16215413$  & $0.13595839$  \\
        $2$ & $0.06548904$  & $0.05488133$  \\
        $3$ & $0.02646573$  & $0.02213681$  \\
        $4$ & $0.01072054$  & $0.00890394$  \\
        $5$ & $0.00438021$  & $0.00354367$  \\
        $6$ & $0.00184589$  & $0.00135357$  \\
        $7$ & $0.00086124$  & $0.00043062$  \\
        \hline
    \end{tabular}
    \caption{An approximately optimal solution to the finite LP with $\kappa = \frac{3}{2}$, $\kmax = 7$, and $\gamma = \frac{1}{3\sqrt{3}}$, rounded to the $8$-th digit after the decimal point, whose $\Gamma \approx 0.51461$.}
    \label{tab:lp-solution-improved}
\end{table}

\section{Conclusions}
\label{sec:conclusions}

This article continues the work of \citet{HuangT/arXiv/2019}, and presents a simplified version of the edge-weighted online bipartite matching algorithm by \citet{Zadimoghaddam/arXiv/2017}, under the online primal dual framework.
Together with \citet{HuangT/arXiv/2019}, we initiate the study of the OCS, an algorithmic ingredient that may find further applications in other online algorithm problems. 

To this end, the independent rounding is a $0$-OCS.
\citet{HuangT/arXiv/2019} gives a $\frac{1}{16}$-OCS, distilled from the algorithm by \citet{Zadimoghaddam/arXiv/2017};
they further show that $1$-OCS is impossible.
This article introduces an improved $\frac{1}{3\sqrt{3}}$-OCS.
Designing a $\gamma$-OCS with the largest possible $0 < \gamma < 1$ is an interesting open problem on its own, and may lead to a better competitive ratio for the edge-weighted online bipartite matching problem.

On the other hand, \citet{HuangT/arXiv/2019} show that even with an imaginary $1$-OCS, the competitive ratio from the current approach is at best $\frac{5}{9}$.
Therefore, we need fundamentally new ideas in order to come closer to the optimal $1 - \frac{1}{e}$ ratio in the unweighted and vertex-weighted cases.
One potential approach is to consider an OCS that allows more than two candidates in each round, which we will refer to as a multiway OCS.
We will leave this as another interesting open problem.

\bibliographystyle{plainnat}
\bibliography{matching.bib}

\begin{thebibliography}{13}
\providecommand{\natexlab}[1]{#1}
\providecommand{\url}[1]{\texttt{#1}}
\expandafter\ifx\csname urlstyle\endcsname\relax
  \providecommand{\doi}[1]{doi: #1}\else
  \providecommand{\doi}{doi: \begingroup \urlstyle{rm}\Url}\fi

\bibitem[Aggarwal et~al.(2011)Aggarwal, Goel, Karande, and
  Mehta]{AggarwalGKM/SODA/2011}
Gagan Aggarwal, Gagan Goel, Chinmay Karande, and Aranyak Mehta.
\newblock Online vertex-weighted bipartite matching and single-bid budgeted
  allocations.
\newblock In \emph{Proceedings of the 22nd Annual ACM-SIAM Symposium on
  Discrete Algorithms}, pages 1253--1264. SIAM, 2011.

\bibitem[Buchbinder et~al.(2007)Buchbinder, Jain, and
  Naor]{BuchbinderJN/ESA/2007}
Niv Buchbinder, Kamal Jain, and Joseph~Seffi Naor.
\newblock Online primal-dual algorithms for maximizing ad-auctions revenue.
\newblock In \emph{European Symposium on Algorithms}, pages 253--264. Springer,
  2007.

\bibitem[Devanur and Jain(2012)]{DevanurJ/STOC/2012}
Nikhil~R Devanur and Kamal Jain.
\newblock Online matching with concave returns.
\newblock In \emph{Proceedings of the 44th Annual ACM Symposium on Theory of
  Computing}, pages 137--144. ACM, 2012.

\bibitem[Devanur et~al.(2013)Devanur, Jain, and Kleinberg]{DevanurJK/SODA/2013}
Nikhil~R Devanur, Kamal Jain, and Robert~D Kleinberg.
\newblock Randomized primal-dual analysis of ranking for online bipartite
  matching.
\newblock In \emph{Proceedings of the 24th Annual ACM-SIAM Symposium on
  Discrete Algorithms}, pages 101--107. SIAM, 2013.

\bibitem[Devanur et~al.(2016)Devanur, Huang, Korula, Mirrokni, and
  Yan]{DevanurHKMY/TEAC/2016}
Nikhil~R Devanur, Zhiyi Huang, Nitish Korula, Vahab~S Mirrokni, and Qiqi Yan.
\newblock Whole-page optimization and submodular welfare maximization with
  online bidders.
\newblock \emph{ACM Transactions on Economics and Computation (TEAC)},
  4\penalty0 (3):\penalty0 14, 2016.

\bibitem[Feldman et~al.(2009)Feldman, Korula, Mirrokni, Muthukrishnan, and
  P{\'a}l]{FeldmanKMMP/WINE/2009}
Jon Feldman, Nitish Korula, Vahab Mirrokni, S.~Muthukrishnan, and Martin
  P{\'a}l.
\newblock Online ad assignment with free disposal.
\newblock In \emph{Proceedings of the 5th International Workshop on Internet
  and Network Economics}, pages 374--385. Springer Berlin Heidelberg, 2009.

\bibitem[Huang and Tao(2019)]{HuangT/arXiv/2019}
Zhiyi Huang and Runzhou Tao.
\newblock Understanding {Z}adimoghaddam’s edge-weighted online matching
  algorithm: unweighted case.
\newblock \emph{arXiv preprint arXiv:1910.02569}, 2019.

\bibitem[Huang et~al.(2018{\natexlab{a}})Huang, Kang, Tang, Wu, Zhang, and
  Zhu]{HuangKTWZZ/STOC/2018}
Zhiyi Huang, Ning Kang, Zhihao~Gavin Tang, Xiaowei Wu, Yuhao Zhang, and Xue
  Zhu.
\newblock How to match when all vertices arrive online.
\newblock In \emph{Proceedings of the 50th Annual ACM Symposium on Theory of
  Computing}, pages 17--29. ACM, 2018{\natexlab{a}}.

\bibitem[Huang et~al.(2018{\natexlab{b}})Huang, Tang, Wu, and
  Zhang]{HuangTWZ/ICALP/2018}
Zhiyi Huang, Zhihao~Gavin Tang, Xiaowei Wu, and Yuhao Zhang.
\newblock Online vertex-weighted bipartite matching: beating $1-\frac{1}{e}$
  with random arrivals.
\newblock In \emph{Proceedings of the 45th International Colloquium on
  Automata, Languages, and Programming}. Schloss Dagstuhl-Leibniz-Zentrum fuer
  Informatik, 2018{\natexlab{b}}.

\bibitem[Huang et~al.(2019)Huang, Peng, Tang, Tao, Wu, and
  Zhang]{HuangPTTWZ/SODA/2019}
Zhiyi Huang, Binghui Peng, Zhihao~Gavin Tang, Runzhou Tao, Xiaowei Wu, and
  Yuhao Zhang.
\newblock Tight competitive ratios of classic matching algorithms in the fully
  online model.
\newblock In \emph{Proceedings of the 30th Annual ACM-SIAM Symposium on
  Discrete Algorithms}, pages 2875--2886. SIAM, 2019.

\bibitem[Karp et~al.(1990)Karp, Vazirani, and Vazirani]{KarpVV/STOC/1990}
Richard~M Karp, Umesh~V Vazirani, and Vijay~V Vazirani.
\newblock An optimal algorithm for on-line bipartite matching.
\newblock In \emph{Proceedings of the 22nd ACM Symposium on Theory of
  Computing}, pages 352--358. ACM, 1990.

\bibitem[Wang and Wong(2015)]{WangW/ICALP/2015}
Yajun Wang and Sam Chiu-wai Wong.
\newblock Two-sided online bipartite matching and vertex cover: Beating the
  greedy algorithm.
\newblock In \emph{International Colloquium on Automata, Languages, and
  Programming}, pages 1070--1081. Springer, 2015.

\bibitem[Zadimoghaddam(2017)]{Zadimoghaddam/arXiv/2017}
Morteza Zadimoghaddam.
\newblock Online weighted matching: beating the $\frac{1}{2}$ barrier.
\newblock \emph{arXiv preprint arXiv:1704.05384}, 2017.

\end{thebibliography}

\end{document}